\documentclass[11pt]{article}
\pdfoutput=1

\usepackage{base}
\numberwithin{equation}{section}

\title{Rank Bounds and PIT for $\Sigma^3 \Pi \Sigma \Pi^d$ circuits
via a non-linear Edelstein-Kelly theorem
}
\author{
Abhibhav Garg
\thanks{University of Waterloo. \ \email{\{a65garg, rafael , aksengup\}@uwaterloo.ca}}
\and
Rafael Oliveira
\samethanks
\and
Akash Kumar Sengupta
\samethanks
}
\date{}

\begin{document}

\maketitle

\begin{abstract}
We prove a non-linear Edelstein-Kelly theorem for polynomials of constant degree, fully settling a stronger form of Conjecture 30 in Gupta (2014), and generalizing the main result of Peleg and Shpilka (STOC 2021) from quadratic polynomials to polynomials of any constant degree.

As a consequence of our result, we obtain constant rank bounds for depth-4 circuits with top fanin 3 and constant bottom fanin (denoted $\Sigma^{3}\Pi\Sigma\Pi^{d}$ circuits) which compute the zero polynomial.
This settles a stronger form of Conjecture 1 in Gupta (2014) when $k=3$, for any constant degree bound; additionally this also makes progress on Conjecture 28 in Beecken, Mittmann, and Saxena (Information \& Computation, 2013).
Our rank bounds, when combined with Theorem 2 in Beecken, Mittmann, and Saxena (Information \& Computation, 2013) yield the first deterministic, \emph{polynomial time} PIT algorithm for $\Sigma^{3}\Pi\Sigma\Pi^{d}$ circuits.
\end{abstract}
\newpage
\microtypesetup{protrusion=false}
\tableofcontents
\microtypesetup{protrusion=true}
\newpage

\ifdefined\DRAFTSUBMISSION
    \linenumbers
\else
\fi

\section{Introduction}\label{section: introduction}

Polynomial Identity Testing (PIT) is a fundamental problem in algebraic complexity theory, with far ranging applications in theoretical computer science as well as in mathematics.
The PIT problem asks whether a given algebraic circuit over a field $\bF$ formally computes the zero polynomial.
When the field $\bF$ is sufficiently large, the PIT problem equivalently asks whether the given algebraic circuit evaluates to zero on all inputs.
This latter perspective is the source for a very simple randomized algorithm for the PIT problem, via the polynomial identity lemma.
When the field is large enough, to test whether the given circuit computes the zero polynomial, one simply evaluates the given circuit at a random point.
This simple algorithm places PIT in the class $\coRP$.

The task of derandomizing the PIT problem is closely tied to the quest of proving explicit lower bounds for algebraic circuits \cite{HS80, Agr05, KI04}. 
Additionally, derandomizing the PIT problem, even for special classes of circuits, has resulted in the derandomization of key problems in mathematics and computer science \cite{AKS04, FS13, Mul17, FGT19}. 
For an overview on PIT, see \cite{Sax09, SY10, Sax14}.

Given the strong consequences arising from the derandomization of the PIT problem, even for special classes of circuits, there has been significant research devoted to its solution.
Due to the difficulty of the general problem, attention has turned to more tractable and natural classes of circuits, such as the class of sparse polynomials (also known as depth-2 circuits) \cite{klivans2001randomness}, and depth-3 circuits \cite{DS07, KS09, SS13}.
The recent breakthroughs in depth reduction \cite{AV08,GKKS13, T15} have demonstrated that the general PIT problem can be reduced to the special cases of the PIT problem for unrestricted depth-3 circuits or homogeneous depth-4 circuits.
These results have rekindled interest in these classes, leading to a concentrated research effort on the PIT problem for these classes, such as the works \cite{KS09, beecken2013algebraic, SS13, gupta2014algebraic, guo2021, DDS21}.

Since the depth-3 and depth-4 cases are equivalent to the general PIT problem, an additional restriction has been considered in the aforementioned works: in the depth-3 case the top fan-in of the circuit is assumed to be constant, and in the depth-4 case both top and bottom fan-ins of the circuits have been assumed to be constant. 
Henceforth, we denote by $\SPS{k}$ the set of all depth-3 circuits with top fan-in $k$ and by $\SPSP{k}{d}$ the set of all depth-4 circuits with top fan-in bounded by $k$ and bottom fan-in bounded by $d$ (where one should think of $k$ and $d$ as being constants).

The work \cite{DS07}, which initiated the study of the PIT problem for $\SPS{k}$ circuits, gave a deterministic, quasipolynomial time algorithm for this problem.
Moreover, in~\cite[Section 7]{DS07}, the authors raised the connection between $\SPS{3}$ identities and a generalization of the famous Sylvester-Gallai problem, thereby initiating a connection between PIT and discrete geometry.
Such a connection was built in the hopes of obtaining deterministic, \emph{polynomial time} algorithms for the PIT problem for $\SPS{k}$ circuits.
As it turns out, the conjecture posed in \cite[Section 7]{DS07} had already been settled (when the base field is $\bR$) by Edelstein and Kelly in \cite{edelstein1966bisecants}.

The above connection was further studied and generalized to $\SPS{k}$ circuits by the works \cite{KS09,SS13}, culminating with a strong and elegant relationship between higher dimensional Sylvester-Gallai configurations and $\SPS{k}$ identities (\cite[Theorem 1.4]{SS13}).

The work of \cite{gupta2014algebraic} proposed non-linear generalizations of several relevant Sylvester-Gallai type conjectures that would be needed in order to generalize to the $\SPSP{k}{d}$ case the connections between PIT and discrete geometry that were forged in the $\SPS{k}$ case.
Following the breakthrough of \cite{S20}, which established the first non-linear Sylvester-Gallai theorem for quadratic polynomials, a sequence of works has confirmed several of Gupta's conjectures \cite{PS20a, PS20b, OS22, OS24, GOPS23, GOS24}.
Among these works, the paper of Peleg and Shpilka \cite{PS20b} has been the only of such works to establish a deterministic PIT algorithm for $\SPSP{3}{2}$ circuits, via a generalization of the Edelstein-Kelly theorem for quadratic polynomials.

In this work, we prove a generalized version of the Edelstein-Kelly theorem for polynomials of bounded degree, thereby fully settling a stronger form of \cite[Conjecture 30]{gupta2014algebraic} and settling Gupta's main conjecture (\cite[Conjecture 1]{gupta2014algebraic}) when $k=3$, for any constant degree $d$.
Our main result, when combined with \cite[Theorem 2]{beecken2013algebraic}, yields the first polynomial-time, deterministic PIT algorithm for $\SPSP{3}{d}$ circuits.

Before formally stating our main result, and its consequences to the PIT problem, we explain how Edelstein-Kelly configurations (and their generalizations) naturally appear in the PIT problem for $\SPSP{3}{d}$ circuits.

\subsection{PIT and generalized Edelstein-Kelly configurations}

We begin by introducing the classical linear Edelstein-Kelly configurations \cite{edelstein1966bisecants}:

\begin{definition}[Edelstein-Kelly configurations]\label{definition: linear EK}
Let $\bK$ be a field, $\cA := \{u_1, \dots, u_a\}, \cB := \{v_1, \dots, v_b\}$ and $\cC := \{w_1, \dots, w_c\}$ be pairwise disjoint subsets of $\bP(\bK^N)$.
We say that $(\cA, \cB, \cC)$ forms a linear Edelstein-Kelly configuration if given any pair of vectors in distinct sets, their linear span contains a vector in the third set.
The rank of an Edelstein-Kelly configuration $(\cA, \cB, \cC)$ is given by $\dim \Kspan{\cA \cup \cB \cup \cC}$. 
\end{definition}

When $\bK = \bR$, \cite[Theorem 3]{edelstein1966bisecants} proves that the rank of any Edelstein-Kelly configuration is at most $3$.
When $\bK = \bC$, the fractional Sylvester-Gallai theorem of \cite{BDWY11} implies that the rank of any Edelstein-Kelly configuration is bounded from above by a universal constant.
The fact that such configurations must be constant dimensional is known as an Edelstein-Kelly type theorem.

With the above definition and rank bound at hand, we explore how Edelstein-Kelly configurations naturally arise in $\SPS{3}$ identities. 
Let $\vx = (x_1, \dots, x_N)$ be a tuple of variables, and let $S := \mathbb{C}[\vx]$ denote the polynomial ring over $\mathbb{C}$. 
In the introduction, we will assume that our given polynomials and circuits are homogeneous, and we follow standard notation and refer to homogeneous polynomials as \emph{forms}.
In \cref{sec: circuits PIT} we will formally show our rank bounds without this assumption.

Suppose a $\SPS{3}$ circuit computes a form $P$, which takes the following form:

\begin{equation}\label{equation: intro depth 3 polynomial}
P = \prod_{i=1}^m \ell_{i}(\vx) + \prod_{j=1}^m g_{j}(\vx) + \prod_{k=1}^m h_{k}(\vx),    
\end{equation}
where each $\ell_i, g_j, h_k$ is a linear form.\footnote{Note that, since we are only given the circuit computing $P$, we do not explicitly know its coefficients and monomials.}

If $P \equiv 0$, that is, the circuit forms an identity (equivalently, computes the zero polynomial), and if this identity is \emph{efficiently represented} -- meaning no subset of the summands add to zero, and the summands share no common factor\footnote{This corresponds to the circuit being \emph{simple} and \emph{minimal} as defined in \cite{DS07}.} -- we are led to ask whether the involved linear forms must necessarily lie in a low-dimensional space, i.e., depend on only a few variables.
To capture the "true number of variables" of the circuit given by \cref{equation: intro depth 3 polynomial}, Dvir and Shpilka~\cite{DS07} introduced the \emph{rank} of a depth-3 circuit as the dimension of the linear span of all linear forms appearing in it. 
For the circuit in \cref{equation: intro depth 3 polynomial}, the rank is given by
$\dim \Cspan{\ell_i, g_j, h_k}_{i,j,k \in [m]}.$

Consider a linear form $\ell_{i}$ from the first gate and a second linear form $g_{j}$ from the second gate.
Since $P \equiv 0$, we have that for any $\va \in \bC^N$ such that  $\ell_i(\va) = g_j(\va) = 0$, it must be the case that  $\prod h_{k}(\va) = 0$.
In other words, we have that $V(\ell_i, g_j) \subset V\br{\prod h_k} = \bigcup_k V\br{h_k}$.
Since the algebraic set $V\br{\ell_i, g_j}$ is irreducible, we must have $V\br{\ell_{i}, g_{j}} \subset V\br{h_{a}}$ for some $a \in [m]$.
This last condition (combined with the symmetry among the gates) is exactly the local constraint arising from the dual formulation of Edelstein-Kelly configurations in \cref{definition: linear EK}.
By the Edelstein-Kelly theorems above we deduce that $\Cspan{\ell_{i}, g_{j}, h_{k}} = O(1)$.
This shows that any $\Sigma^3\Pi\Sigma$ identity essentially depends on constantly many variables.
Combined with \cite{KS08}, this gives a black box deterministic PIT algorithm for $\SPS{3}$ circuits.

Let us now see how a natural non-linear generalization of \cref{definition: linear EK} arises in the study of $\SPSP{3}{d}$ PIT.
Consider a form $Q$ computed by a $\Sigma^{3} \Pi \Sigma \Pi^{d}$ circuit.
It has the form
$$ Q = \prod_{i=1}^{m_1} A_{i}(\vx) + \prod_{j=1}^{m_2} B_{j}(\vx) + \prod_{k=1}^{m_3} C_{k}(\vx),$$
where $A_{i}, B_{j}, C_{k}$ are forms of degree at most $d$.
If $Q \equiv 0$ and the representation is efficient, as in the previous case, we have $V(A_i, B_j) \subseteq V\br{\prod C_k}$.\footnote{By symmetry among the gates, we also have $V(A_i, C_k) \subseteq V\br{\prod B_j}$ and $V(B_j, C_k) \subseteq V\br{\prod A_i}$.}
However, as the forms are not necessarily linear, we have that $V\br{A_i, B_j}$ is not necessarily irreducible, and thus we cannot guarantee the existence of $k \in [m]$ such that $V\br{A_i, B_j} \subset V\br{C_k}$.
Nevertheless, the above relations led Gupta \cite{gupta2014algebraic} to generalize the definition of Edelstein-Kelly configurations in the following way:

\begin{definition}[Non-linear Edelstein-Kelly configurations]\label{definition: non-linear EK simple intro}
    Let $\bK$ be a field, $\cA, \cB$ and $\cC$ be finite sets of irreducible forms in $\bK[\vx]$ of degree at most $d$. 
    We say that $(\cA, \cB, \cC)$ forms a $d$-Edelstein-Kelly configuration over $\bK$ if the following conditions hold:
    \begin{enumerate}
        \item any two forms are pairwise non-associate
        \item for any $A \in \cA, B \in \cB$, we have that:
        $$ \prod_{C \in \cC} C \in \radideal{A, B}. $$
        Moreover, such relation works for any permutation of the sets $\cA, \cB$ and $\cC$.
    \end{enumerate}
\end{definition}

\noindent Note that when $d=1$, the above definition, together with the fact that ideals generated by linear forms are prime, becomes the same as the usual (linear) Edelstein-Kelly configuration, recovering \cref{definition: linear EK}.
Moreover, by the algebra-geometry correspondence given by the Nullstellensatz, whenever the field $\bK$ is algebraically closed, the above algebraic condition becomes $V\br{A, B} \subseteq V\br{\prod_{C \in \cC} C}$, thereby recovering the geometric condition from $\SPSP{3}{d}$ identities. 

Similarly to the linear case, Gupta proposed the following conjecture:\footnote{\cref{conjecture: gupta main simple intro} is a stronger form of \cite[Conjecture 1]{gupta2014algebraic} when $k=3$, where transcendence degree is replaced by the dimension of the span of the polynomials in the configuration.} 

\begin{conjecture}[Non-Linear Edelstein-Kelly conjecture]\label{conjecture: gupta main simple intro}
    There exists a function $\lambda : \bN \to \bN$ such that, if $(\cA, \cB, \cC)$ forms a $d$-Edelstein-Kelly configuration over a field $\bK$ of characteristic zero, then 
    $$ \dim \Kspan{\cA \cup \cB \cup \cC} \leq \lambda(d).$$
\end{conjecture}

Since every $\SPSP{3}{d}$ identity gives rise to a non-linear Edelstein-Kelly configuration (the three sets are simply the irreducible factors of each of the product gates), the above conjecture, when combined with \cite[Theorem 2]{beecken2013algebraic}, would yield the first \emph{polynomial-time}, deterministic PIT algorithm for this circuit class.

The work of Peleg and Shpilka~\cite[Theorem 1.6]{PS20b} confirmed \cref{conjecture: gupta main simple intro} for the case when $d = 2$.
Our main theorem is to confirm the above conjecture for any value of $d$.

\subsection{Our Results}\label{subsection: results}

Now that we have discussed the connections between Edelstein-Kelly configurations and PIT, we are ready to state our main result: non-linear Edelstein-Kelly configurations have bounded rank.

\begin{restatable}[Rank bound for EK-configurations]{theorem}{ekmain}\label{theorem: EK main}
    There exists a function $\lambda: \bN \rightarrow \bN$ such that for any $d$-Edelstein–Kelly configuration $(\cA, \cB, \cC)$ over a field $\bK$ of characteristic $0$, we have 
    \[\dim(\Kspan{\cA \cup \cB \cup \cC})\leq \lambda(d).\]
\end{restatable}

Our main theorem generalizes the main result of \cite{PS20b} to forms of any bounded degree $d$.
Additionally, our proof technique is more general and simpler than the proof in \cite{PS20b}, as we have a simpler case analysis than their work.

Since any simple and minimal $\SPSP{3}{d}$ identity gives rise to a non-linear Edelstein-Kelly configuration, we obtain the following corollary.

\begin{restatable}{corollary}{rankbound}\label{theorem: rankbound for identities}
    There is a function $\lambda : \bN \to \bN$ such that for any simple and minimal $\Sigma^3\Pi\Sigma\Pi^d$-identity $\Phi$ over a field $\bK$ of characteristic $0$, we have $\Rank{\Phi} \leq \lambda(d).$ 
\end{restatable}

The above rank bound, when combined with \cite[Theorem 2]{beecken2013algebraic}, implies the following derandomization result for PIT.

\begin{restatable}{corollary}{mainPIT}\label{corollary: PIT}
    There is a deterministic polynomial-time algorithm for identity testing of $\Sigma^3\Pi\Sigma\Pi^d$-circuits.
\end{restatable}

Our proof strategy, much like the proofs in previous works on the non-linear generalizations of Sylvester-Gallai configurations, is to "construct" the function $\lambda$ by induction on the degree of Edelstein-Kelly configurations.
The main tool employed in previous works to reduce the degree of these configurations is to apply a sequence of general quotients (see definition in \cref{sec:generalquot}) to certain special vector spaces of forms that appear in our configuration.
Since these general quotients map polynomial rings into quotient rings of polynomial rings, we need to further generalize our definition of Edelstein-Kelly configurations to this more general setting.
We provide this generalization in \cref{def:ek}, and in \cref{subsection: contributions and comparison} we outline the subtleties and challenges we need to overcome to achieve our results.

Equipped with this general and more versatile type of configurations, our strategy to reduce the degree of the configuration is similar to the one employed in \cite{GOS24}: we find small vector spaces of forms that, when quotiented out via a general quotient, makes a constant fraction of the forms in our Edelstein-Kelly configuration to "factor more," and therefore drop degree.
However, due to the combinatorial restrictions imposed by Edelstein-Kelly configurations, the implementation of this step is more delicate, and requires the development of a new potential function which captures the reducibility of an Edelstein-Kelly configuration with respect to a given vector space of forms.
This connection is established in \cref{subsection: potential function}, and we heavily use it in \cref{section: EK theorem} to prove our main technical theorem, which we now state.

\begin{restatable}{theorem}{generalekmain}\label{theorem: general ek main}
    For any two positive integers $d,e$ such that $e\geq 2$ and $d\leq e$ , there exist ascending functions $\Lambda_{d,e}: \bN^{e} \to \bN^{e}$ and  $\lambda_{d, e}:\bN \rightarrow \bN$, both independent of $\bK$ and $N$, such that the following holds. 

Let $\bK$ be an algebraically closed field of characteristic $0$ and $S := \bK[x_1,\cdots,x_N]$. 
Let $U\subset S_{\leq e}$ be a $\Lambda_{d,e}$-strong graded vector space and $R=S/(U)$. 
If $(\cA, \cB, \cC)$ is a  $(d, z, R)$-EK configuration in $R$ for some $z \in S_{1}$, then we have
$$\dim (\Kspan{\cA \cup \cB \cup \cC} )\leq \lambda_{d,e}(\dim(U)).$$
i.e. the dimension of the $\bK$-linear span of $\cA, \cB, \cC$ is upper bounded by a function of $d,e,\delta$ and $\dim(U)$, which is independent of the field $\bK$, the number of variables $N$ and the cardinality of $\cA, \cB, \cC$.
\end{restatable}

In \cref{section: EK theorem}, we show how \cref{theorem: EK main} follows as an easy corollary of the above theorem.

\subsection{Proof overview}

We now describe our proof strategy and its implementation, giving an overview of the new ideas needed and the new technical challenges that we overcome to implement our inductive approach and prove our generalized Edelstein-Kelly theorem.

At a high level, our approach to prove rank bounds on Edelstein-Kelly (EK) configurations follows the inductive approaches from previous works, in particular the works \cite{PS20b, OS24, GOS24}.
Their main idea is to prove the existence of small vector spaces of forms which are in some sense "present in many forms" in the given configuration.

Let $R := \bC\bs{x_{1}, \dots, x_N}$, and $S := R\bs{y_{1}, \dots, y_n}$ be two polynomial rings and $\cF := (\cA, \cB, \cC) \subset S$ be a $d$-Edelstein-Kelly configuration.
Thus, we can write $\cF = \cF_1 \cup \cdots \cup \cF_d$, where $\cF_e := (\cA_e, \cB_e, \cC_e)$ are the forms in $\cF$ of degree $e$ (with their place in the partition).
We start with the assumption that each form only depends on constantly many variables in $S$.\footnote{The set of variables may be different across the forms, otherwise the main theorem is trivially true.
Also, note that the variables of $S$ are both the $x$ and $y$ variables.}
A number of key ideas can already be highlighted in this easier setting.
 
We want to control the highest degree forms $\cF_{d}$ in our configuration, with the goal of reducing to the case when the highest degree is $d-1$, where we can proceed inductively.
From now on whenever we talk about conditions on two of the sets, we assume that the symmetric conditions with the permuted sets also hold, unless stated otherwise.

If many ideals of the form $(A_i, B_j)$, where $A_i \in \cA_d, B_j \in \cB_d$ are prime, then the Edelstein-Kelly condition implies that the set $\cF_d$ is essentially a (robust) linear Edelstein-Kelly configuration.
In this case, the linear Sylvester-Gallai theorems imply that constantly many forms $F_{1}, \dots, F_{a}$ are a basis for $\Cspan{\cF_{d}}$, and if $z_{1}, \dots, z_{r}$ is the union of the set of variables of $F_{1}, \dots, F_{a}$, then $\cF_{d} \subset \bC\bs{z_{1}, \dots, z_{r}}$. 
This is sufficient to control the forms of $\cF_{d}$ and apply our inductive step.
The interesting case is thus when there are many ideals generated by forms in $\cF_d$ that are not prime. 
In this case, the goal is to show that the forms in $\cF_d$ must share many variables in common.
The steps listed below show how previous works have dealt with this case.

\paragraph{Step 1 - Structure theorems:} In this work, the main structure theorem we use is \cite[Theorem~4.16]{GOS24}.
In fact we require a stronger form of this theorem, which we prove in \cref{lem:absredcount} and \cref{cor:strongprimebound}.
We briefly explain how such a structure theorem is helpful in our setting.

Suppose we have an irreducible form $P \in S \setminus \ideal{x_{1}, \dots, x_{N}}$.
Suppose $Q_{1}, \dots, Q_{r} \in R$ are irreducible forms.
We are interested in bounding the number of $Q_{i}$ such that $\ideal{P, Q_{i}}$ is not prime.
In \cite{OS24}, a bound was proved for the number of $Q_{i}$ such that $\ideal{P, Q_{i}}$ is not radical, and this bound only depends on the degree of $P$.
In the case of prime ideals however, things are more complicated.
Consider for example $P = y_{1}^{4} - x_{1} x_{2} y_{2}^{2}$.
The form $P$ is irreducible, since it is linear in $x_{1}$.
However, for any $Q \in R$, the ideal $\ideal{P, x_{1} x_{2} - Q^{2}}$ is not prime, since $y_{1}^{2} - Qy_{2} \in \ideal{P, x_{1} x_{2} - Q^{2}}$.
Therefore, such a bound cannot exist for every $P$.

The structure theorem \cite[Theorem~4.16]{GOS24} states that the above only happened because $P$ is reducible as a polynomial in the ring $\overline{\bC\br{x_{1}, \dots, x_N}}\bs{y_{1}, \dots, y_{n}}$, since it has factorisation $P = \br{y_{1}^{2} - \sqrt{x_{1} x_{2}} y_{2}} \br{y_{1}^{2} + \sqrt{x_{1} x_{2}} y_{2}}$.
In particular, the theorem states that as long as $P$ is irreducible not just as a polynomial in $S$, but also as a polynomial in $\overline{\bC\br{x_{1}, \dots, x_N}}\bs{y_{1}, \dots, y_{n}}$, then the number of $Q_{i}$ such that $\ideal{P, Q_{i}}$ is not prime is bounded by a function of only $\deg{P}$.
Such a $P$ is said to be absolutely irreducible over the variables $x_{1}, \dots, x_{N}$.

The reason why this type of structure theorem is useful is the following: forms that are absolutely reducible will factor when a graded quotient is applied.
Therefore, we want to try and find a \emph{small vector space} of linear forms such that many forms in our configuration are absolutely reducible with respect to this vector space.
Due to the additional constraints coming from the EK configurations, we need the more robust form of \cite[Theorem 4.16]{GOS24} for pencils of forms, which we prove in \cref{lem:absredcount}.

The proof of \cref{lem:absredcount} proceeds by studying the Noether equations that define reducibility over algebraically closed fields.
If a form $P$ is absolutely irreducible, then it does not satisfy some Noether equation, and therefore neither does $P + Q_{i}$ for all except constantly many $Q_{i}$.
As in the proof of \cite[Theorem~4.16]{GOS24}, the non trivial step here is to show that the bound is independent of $n, N$.
We obtain this using a Bertini type theorem.

\paragraph{Step 2 - Finding a set of variables which makes several forms "factor more":}

We start with the configuration $\cF$ and vector space $V := 0$.
The vector space $V$ will play the role of the small set of variables that control $\cF$ (i.e., make many forms in $\cF$ "factor more").
As discussed before, if for every $A \in \cA_{d}$, many of the ideals $\ideal{A, B}$ with $B \in \cB_{d}$ are prime, then the set of degree $d$ forms $\cF_{d}$ is a robust linear EK configuration, therefore it has small linear rank.

Therefore, let us assume that this does not happen.
In other words, for a constant $p$ we can essentially assume that $A_{1}, \dots, A_{p} \in \cA$ are such that $\ideal{A_{i}, B_{j}}$ is not prime for every $1 \leq i \leq p$ and $1 \leq j \leq \abs{\cB} / 2$.
As discussed above, the structure theorem will imply that the forms $B_{j}$ are all absolutely reducible over the variables of $A_{1}, \dots, A_{p}$.
Since $p$ is constant, and we are assuming each form in $\cF$ depends on constantly many variables, the total number of variables of $A_{1}, \dots, A_{p}$ is a constant.
Therefore when we update $V$ to include the space spanned by these variables we make a some progress, since the $B_j$'s are factoring under a general quotient.

Now that half the forms in $\cB_{d}$ are absolutely reducible, we try to control $\cA_{d}$, in particular the forms that are still absolutely irreducible over the new $V$.
This forces us to look at the EK relationships between $A_{i}$ and $B_{j}$ for $1 \leq j \leq \abs{\cB} / 2$.
Recall that these forms of $\cB$ are already absolutely reducible over $V$.
This makes it a challenge to study the ideals $\ideal{A_{i}, B_{j}}$, since we cannot directly apply the prime bound.
Hence, we have to apply the structure to the factors of $B_{j}$ as a polynomial with coefficients in $\overline{\bC\br{V}}$, which poses some algebraic challenges (and will force us to work over the general quotient).

In order to control $\cA_d$, we look for long sequences $A_{1}, \dots, A_{r}$ of forms that are not only absolutely irreducible over $V$, but also their pairwise spans $\Cspan{A_i, A_j}$ does not contain any absolutely reducible form (we call such sequences unbreakable pencils \cref{def: unbreakable pencil}).
This is where our strengthened structure theorem applies: if no such long sequences exist, then the forms in $\cA_{d}$ pairwise span absolutely reducible forms, and by adding a few more variables to $V$, we can deduce that every form in $\cA_{d}$ is absolutely reducible.
We are thus left with the case that long sequences $A_{1}, \dots, A_{r}$ do exist (here, long means bigger than the bounds of the structure theorem, therefore $r$ is still a constant).
In this case let $V'$ be the vector space obtained by adding the variables of $A_{1}, \dots, A_{r}$ to $V$.

With the above long sequence at hand, we are able to deduce that a large fraction of the forms in $\cB$ that are already absolutely reducible over $V$, will be be \emph{even more absolutely reducible} over $V'$.
This notion of being even more absolutely reducible is captured by a more refined potential function than used in previous works, in particular because the potential function will depend on the general quotient used. 
This poses some algebraic challenges, as this restricts the generality of the graded quotients that we can apply.
We explain how we overcome this in Step 3 below.

If $\cB$ is not the smallest among the three sets, then after constantly many iterations of this step, we can argue that the potential function reaches its maximum, which will imply that all the forms in $\cB$ are absolutely reducible.
However, we face the following combinatorial challenge: we cannot force a particular set to be controlled.
The reason being that the change in the potential function at each step depends on the relative sizes of $\cA_{d}, \cB_{d}, \cC_{d}$.
If $\cB$ is the smallest of the three sets, then we need to do some more combinatorial work to be able to control one of the other sets from the fact that we control $\cB$.

If we manage to overcome all of these challenges, then we reach a stage where we have a vector space $V$, and every form in one of the sets $\cA_{d}, \cB_{d}, \cC_{d}$ is absolutely reducible over $V$.
Another application of the structure theorem will allows us to find a slightly bigger vector space $V'$ such that all of $\cF_{d}$ is absolutely reducible over $V'$, and we can apply induction.
We now summarise the above challenges, and discuss how we overcome them.

\paragraph{Combinatorial challenges and solutions:}
As discussed above, the outline of our method is to first show that half the forms in some set are absolutely reducible, and then extend this to control all the forms in some large enough set.
This is difficult to do when the sets are all of \emph{different sizes}.
The main challenge comes from the fact that we have no way of controlling how $\cF_{d}$ fails to be a robust linear EK configuration.
It might be the case that this failure only allows us to control the smallest of the sets.

This is overcome by noting that if the sets are of vastly different sizes, then the failure of configuration to be a linear EK configuration will be more structured, and we will be able to control one of the bigger two sets.
However if the sets are not that different in size, then control of the smallest gate will be enough to make the above arguments work, with slight modifications in parameters.
Our intermediate lemma become slightly more involved since they will now depend on these relative sizes.

\paragraph{Algebraic challenges and solutions:} 
On the algebraic side, our structure theorem is about absolutely irreducible form, but as we increase the base vector space of linear forms, we need to be able to apply the structure theorem to argue about forms that are absolutely reducible.
We are able to still use the structure theorem by arguing about the factors of the images of reducible forms under the graded quotients.
This in turn will require us to compare images of forms under sequences of graded quotients, so we will always have to ensure that our quotients are \emph{compatible} with previous quotients.
To this end, we introduce the notion of twisted graded quotients, which allow us to \emph{track the progress on the factorization of the forms in our original configuration}.

Further, we will have to define a \emph{potential function} that captures \underline{fine grained} notions of "how close a form is to an algebra," and also how much it factorises with respect to that vector space.
We have to also study how this potential changes as the vector space is iteratively modified.
This potential function will also be defined using graded quotients, and tracking the change in potential as the vector space is modified will require the quotients to be compatible as described above.

\paragraph{Step 3 - Graded quotients and potential function:}
We can now finally describe our algebraic tools and our technical contributions on the algebraic side.
The notion of graded quotients was introduced in \cite{OS24}, and further developed in \cite{GOS24}, as a generalization of projection maps used in \cite{S20, PS20a, PS20b, PS22, GOPS23}.
In the simplest case, graded quotients can be described as follows. Let $S:=\bC[x_1,\cdots,x_m,y_1,\cdots,y_n]$, and let $z$ be a new variable. For $(\alpha_1,\cdots,\alpha_m)\in \bC^m$, consider the map $\varphi_\alpha:S[z]\rightarrow \bC[z,y_1,\cdots,y_n]$ given by $x_i\mapsto \alpha_i z$. Given a finite set of forms $\cF\subseteq S$, one can apply such a map $\varphi_\alpha$ and obtain a corresponding set of forms $\varphi_\alpha(\cF)$ in the new polynomial ring. If $\alpha$ is chosen to be general enough, then the map $\varphi_\alpha$ preserves several useful properties of forms, such as linear independence. Moreover, if the image $\varphi_\alpha(\cF)$ is low-dimensional for general choices of $\alpha$, then the span of $\cF$ itself must be low-dimensional. 
These properties played a key role in establishing uniform bounds for quadratic Edelstein-Kelly and Sylvester-Gallai-type configurations in \cite{S20, PS20a, PS20b, PS22, GOS22, GOPS23}.

In \cite{OS24}, a more general tool of graded quotients was developed which can deal with higher degree forms instead of variables $x_i$. In particular, for a $\bC$-algebra $R$ and a finite dimensional vector space $V\subseteq R$, with basis $F_1,\cdots,F_m$, a graded quotient is given by a map of the form $F_i\mapsto \alpha_iz^{\deg(F_i)}$. In particular, we have a map $\varphi_{V,\alpha}:R[z]\rightarrow R'$, where $R'$ is a quotient ring $R[z]/I$. Moreover, if the forms $F_i$ are ``strong" enough, i.e. they behave as variables, then useful properties of $\varphi_\alpha$ also hold in this general setting as shown in \cite{OS24}. In \cite{GOS24}, it was shown that such graded quotients can also capture absolute irreducibility. In the situation of polynomial rings, this property shows that a form $F\in S$ is absolutely irreducible with respect to $\bC(x_1,\cdot,x_m)$, iff $\varphi_\alpha(F)$ is irreducible for general choices of $\alpha$. 

Recall that in our setting of EK-configurations, we want to successively build vector spaces $V$, such as $V=\Cspan{x_1,\cdots,x_m}$, which make more number of forms absolutely reducible with respect to the function field $\bC(V)$. Note that absolute irreducibility over $\bC(x_1,\cdots,x_m)$ is significantly more challenging to handle than irreducibility in the polynomial ring $\bC[z,y_1,\cdots,y_m]$. Therefore, these graded quotients provide an effective tool to handle  absolute irreducibility in our case. However, this approach leads us beyond the realm of polynomial rings, as the graded quotient rings are not necessarily polynomial rings. Thus, we consider EK-configurations more generally in quotient rings, and prove rank bounds in this general setting in \cref{theorem: general ek main}.

Moreover, in \cref{prop: graded quot of ek}, we show that the Edelstein-Kelly property is preserved under such general quotients. This property enables us to inductively reduce the degree of EK-configurations by controlling the highest degree forms in $\cF$. As mentioned before, several new technical challenges appear in our setting, which were not present in the previous works that employed graded quotients.

\emph{Composition.} Given a vector space $V$, we would like to increase it to another vector space $Y$, such that more forms of our EK-configuration become absolutely reducible with respect to the corresponding fields of rational functions. Therefore, by the equivalent criterion discussed above, we want to find new forms $F$ such that $\varphi_{Y,\beta}(F)$ is reducible while $\varphi_{V,\alpha}(F)$ is irreducible. Although $V\subseteq Y$, in general these two quotient maps might not be compatible. In other words, the quotient map $\varphi_{Y,\beta}$ might not be a composition of two quotients of the form $\varphi_{V,\alpha}$ and $\varphi_{W,\gamma}$. In \cref{sec:generalquot}, we define the notion twisted quotients and overcome this technical obstacle by showing that we can still decompose $\varphi_{Y,\beta}$ as a composition of $\varphi_{V,\alpha}$ and an appropriate twisted quotient (up to isomorphism).

\emph{Potential function.} In the setting of EK configurations, it is not enough to increase the number of forms that factor under a quotient. We also need a more refined measure of how much a given form factors. Although previous works such as \cite{OS24, GOS24} had employed graded quotients, this finer requirement is a novel that we face due to EK-configurations. In \cref{subsection: potential function}, we quantify this by defining a potential function using the graded quotients. In particular, for a graded quotient map $\varphi_{V,\alpha}$ we have an associated potential function $\Psi$ such that the following two key properties hold. For two forms $P,Q$, if $\Psi(P)>\Psi(Q)$, then $P$ depends more on the variables in $V$ than $Q$ does. Moreover, $\Psi(P)>\Psi(Q)$ also implies that $P$ factors more with respect to $V$ than $Q$. For a given EK-configuration $\cF$, we quantify the total amount of factoring and closeness with respect to $V$ as the group potential, which is a sum of the individual potentials of forms in $\cF$.

Given these technical tools, we employ the strategy discussed earlier to prove our rank bound on EK-configurations. Moreover, our proof works without the simplifying assumptions made in the beginning. We discuss the key ideas below. 

\paragraph{Step 4 - Removing assumptions:} 
We started with the assumption that every form in $\cF$ depended only on constantly many variables.
In general, it is possible that the forms depend on many or even all the variables.
This issue was overcome in \cite{OS24}, building on the seminal work of \cite{AH20}.
They show that forms $H_{1}, \dots, H_{a}$ of high enough strength (a notion we define in \cref{sec:strongalg}) behave essentially like variables, in the sense that $\bC\bs{H_{1}, \dots, H_{a}}$ is isomorphic to a polynomial ring.
Further, the extension $\bC\bs{H_{1}, \dots, H_{a}} \subset S$ has many of the properties that the extension $\bC\bs{x_{1}, \dots, x_{a}} \subset S$ has, the most useful of them being that this extension preserves arbitrary intersection of ideals.
The work of \cite{OS24} is able to generalize the notion of general projections of linear forms from \cite{S20} to the setting of strong vector spaces, and in our work we build on \cite{OS24} to add the extra flexibility to work with \emph{compatible quotients}, which is done in \cref{sec:generalquot}.

The notion of absolute irreducibility for strong algebras is done in \cite{GOS24}, where they also define the notion of absolute reducibility with respect to strong vector spaces (see \cref{subsec:absirredstrong}).
In that work, they show that general graded quotients send forms that are absolutely reducible with respect to strong vector spaces to reducible forms.
We are able to use the generalization done in \cite{GOS24} and combine it with our approach to solve the Edelstein-Kelly theorem.

\subsection{Summary of contributions \& comparison with previous works}\label{subsection: contributions and comparison}

We are now able to provide a summary of our contributions, along with a more in-depth comparison between the most relevant previous works to ours: \cite{PS20b, GOS24}.
We begin by summarising our contributions:

\begin{enumerate}
    \item We prove that non-linear Edelstein-Kelly configurations have bounded rank, therefore proving that the blackbox PIT algorithm of \cite{beecken2013algebraic} runs in polynomial-time for $\SPSP{3}{d}$ circuits
    \item On the algebro-geometic size, we are able to \emph{significantly enhance} the general graded quotients developed in \cite{OS24} via the compatibility and twists to be able to more carefully track the reducibility of any form with respect to a sequence of compatible quotients, as well as to lift the bounds from the quotiented spaces to the original space.
    This ability, along with our potential function, allows us to get quantitative control over how our configurations factor with respect to a given sequence of strong vector spaces.

    Another important point, albeit less central than the above, is the refinement of the prime bounds of \cite{GOS24} to be able to capture irreducibility inside of certain pencils of forms.
    \item Combinatorially, we provide a more streamlined method to iteratively construct the desired vector spaces, as well as a simpler and more general procedure to find a suitable small vector space to increase the strong vector space under which our configuration factors more.
    This allows us to handle the more restrictive combinatorial constraints imposed by the Edelstein-Kelly configurations. 
\end{enumerate}

\paragraph{Comparison between our work and \cite{PS20b}.} 
A key technical ingredient in the proof of \cite{PS20b} is a structure theorem for ideals generated by pairs of quadratics, that was introduced in \cite{S20} and then further refined in \cite{PS20a} and \cite{PS20b}.
The structure theorem essentially states that the ideal $I = \ideal{Q_{1}, Q_{2}}$ (where $Q_{1}, Q_{2}$ are quadratic forms) is either prime, or has a minimal prime of the form $\ideal{Q_{1}, a}$ for a linear form $a$, or has a linear minimal prime.
This is an extremely fine grained classification of the structure of ideals generated by two quadratics, which one cannot hope to develop in more general settings.

With this theorem at hand, the proof of \cite{PS20b} proceeds by considering all ideals $\ideal{Q_{1}, Q_{2}}$ with $Q_{1}, Q_{2}$ in different sets, and classifying them based on which of the three structural properties they have.
The simpler case is when only the first and third of the above conditions is satisfied.
If every form $Q_{i}$ satisfies the first structural condition with a constant fraction of the other forms, then the configuration is essentially a partial linear EK configuration.
If every form $Q_{i}$ satisfies the third structural condition with a constant fraction of the other forms, then there exists a small collection of linear forms $v_{1}, \dots, v_{r}$ such that every form in the configuration is in the ideal $\ideal{v_{1}, \dots, v_{r}}$, and a general projection can be applied.
In general, it might be the case that other forms satisfy the first condition with constant fraction of the other forms, and some forms satisfy the second condition.
In this case, a careful case analysis is performed to combine the above two arguments.

The hard case is when there are forms that satisfy the second property with most other forms in the configuration.
This boils down to the condition that most forms in the configuration are of the form $Q_{0} + a_{i} b_{i}$ for some quadratic form $Q_{0}$.
When this happens, much work is done to show that the forms $a_{i}, b_{i}$ themselves form a partial EK configuration.
The remaining small fraction of forms that are not of this type also have to be dealt with, leading to a further case analysis.

In comparison, our method gets rid of the fine-grained structure theorem, in favour of what we call the prime bounds \cref{lem:primebound,}, \cref{lem:absredcount}.
We note that these bounds are not new, \cref{lem:primebound} was proved and used in \cite{GOS24}, and a similar bound for radical ideals was proved and used in \cite{OS24}.
There is almost no hope of generalising the earlier structure theorem beyond the case of quadratics or cubics, since the degree of the ideals grow fast (this is what controls the number of cases), as do the complexity of prime ideals of low degree (this is what controls the complexity of each case).
Our method of using the prime bounds not only easily generalises, but also does so without requiring as much case analysis.

Extending the result from quadratics to forms of higher degree also required a stronger inductive hypothesis, and requires us to generalise the notion of EK configurations to rings that are quotients of polynomial rings.
This idea was introduced in \cite{OS24}, and was also used in \cite{GOS24}.

In addition to the generalization to higher degree forms, another advantage of our approach is in the simplification and streamlining of some of the combinatorial aspects of non-linear EK configurations, which allows us to do less case analysis in the more general setting.

\paragraph{Comparison between our work and \cite{GOS24}.}
In this work, we use the main structure theorem from \cite{GOS24}, along with a technical observation which allows us to account for slightly more general primality phenomena. 
This appears in \cref{lem:absredcount}, which informally states that being absolutely irreducible is a fairly robust notion, and is almost invariant under perturbation by change of constant term.
At a high level, this allows us to isolate the highest degree forms in our EK configuration.
This in turn plays a crucial role in the combinatorial arguments used to control the sets in the configuration.

Since the combinatorics of the Edestein-Kelly problem is more involved than product Sylvester-Gallai configurations, our contributions (compared to \cite{GOS24}) are more on the development of a more refined potential function, along with a more refined general graded quotient.
In the analysis of \cite{GOS24}, it was fairly straightforward to lower bound the change in potential in each step of the argument.
Here, getting this bound is more challenging, and requires us to analyze not just one graded quotient, but two graded quotients at a time.
Comparing two graded quotients introduces its own challenges, since the composition of two graded quotients is a slightly different object that just one larger graded quotient.
This requires the introduction of a slightly more general quotient (a twisted graded quotient).

\subsection{Related work}

We now elaborate on previous and related works on generalizations of Sylvester-Gallai configurations and the PIT problem for depth-4 circuits.

\subsubsection{Sylvester-Gallai and Edelstein-Kelly configurations}

In \cite{gupta2014algebraic}, as a first step of his programme to give deterministic, poly-time algorithms for the PIT problem for $\SPSP{k}{d}$ circuits, Gupta proposed the study of an algebro-geometric generalization of the linear Sylvester-Gallai configurations, later termed radical Sylvester-Gallai configurations.

The first progress on Gupta's radical Sylvester-Gallai conjecture (\cite[Conjecture 2]{gupta2014algebraic}) was made by \cite{S20} in the special case when $d = 2$.
Subsequently, \cite{OS22} proved \cite[Conjecture 2]{gupta2014algebraic} for $d = 3$, and \cite{OS24} fully resolved this conjecture in the affirmative for any degree.

However, as we have seen in the connections to the PIT problem, two weaker requirements appear: one is only required that the product of certain forms lies in the radical of any pair of forms, and also the forms are divided into sets.
This led \cite{PS20a} to consider product Sylvester-Gallai configurations, where the first of the two extra conditions is now included. 
In the same paper, the authors prove a product Sylvester-Gallai theorem for quadratic forms.
Recently, the work of \cite{GOS24} has fully resolved the product Sylvester-Gallai conjecture for any degree $d$.

Finally, in \cite{PS20b}, building on their previous works, Peleg and Shpilka settled \cref{conjecture: gupta main simple intro} for the case when $d = 2$, thereby achieving a deterministic, polynomial-time algorithm for the PIT problem for $\Sigma^{3} \Pi \Sigma \Pi^{2}$ circuits.
Our work now resolves \cref{conjecture: gupta main simple intro} for any constant $d$.
In the next subsection, after we discuss our proof overview, we will provide a more in-depth comparison between our work and the work of \cite{PS20b}.

Recent works \cite{PS22, GOS22, GOPS23} have explored robust and higher-codimensional extensions of the radical Sylvester-Gallai theorems. 
While these results address intriguing problems in extremal combinatorial geometry, they are also driven by motivations from the PIT problem for depth-4 circuits. 
The hope is that these broader generalizations could lead to new insights -- or perhaps a more accessible proof -- of Gupta's conjectures.

\subsubsection{PIT for depth-4 circuits}

Depth-4 circuits with bounded top and bottom fan-in are the simplest circuit classes for which we still lack deterministic polynomial-time algorithms for the PIT problem. 
Thus, this model has been attacked via approaches that go beyond the Sylvester-Gallai-based techniques discussed earlier.

In \cite{DDS21}, the authors present a quasipolynomial-time PIT algorithm for $\SPSP{k}{d}$ circuits, leveraging the Jacobian method introduced in \cite{asss16}. 
Their key insight involves applying the logarithmic derivative and its associated power series expansion to transform the top summation gate of the circuit into a powering gate. 
While this transformation technically violates the bounded top fan-in constraint, circuits with powering gates are well studied, and efficient PIT algorithms are known for such models. 
This reduction enables the use of existing tools for PIT in powering-gate circuits, ultimately leading to their quasipolynomial-time result.

The breakthrough work of \cite{LimayeST21} on proving lower bounds for bounded-depth arithmetic circuits offers another route toward derandomizing PIT for this model. 
The \emph{hardness-versus-randomness paradigm} has been a powerful tool for establishing connections between circuit lower bounds and derandomization~\cite{Agr05, KI04}. 
Notably, the results of~\cite{CKS19} extended these tradeoffs to the bounded-depth regime, showing that analogous principles continue to hold even in this restricted setting.
By leveraging these developments, one obtains a \emph{subexponential-time} PIT algorithm for depth-4 circuits. 
Building further on the lower bound techniques of~\cite{LimayeST21}, \cite{AF22} constructed another \emph{hitting set generator} tailored to bounded-depth circuits, yielding an alternative subexponential-time PIT algorithm with improved parameters. 
This construction exploits the fact that such circuits are inherently unable to detect low-rank matrices, due to the algebraic hardness of computing the determinant in low-depth settings.
It is worth emphasizing, however, that neither of these approaches appears capable of yielding a \emph{fully polynomial-time} PIT algorithm for $\Sigma^k \Pi \Sigma \Pi^d$ circuits. 

\subsubsection{Concurrent Work}

In an exciting article concurrent and independent to this one, Guo and Wang \cite{guo2025deterministic} have given a polynomial time deterministic PIT algorithm for $\tspsp$ circuits over fields of any characteristic, under the additional assumption that one of the three gates computes a squarefree polynomial.
The techniques in the two articles are different, and the release of the two articles was coordinated after completion.

\section*{Acknowledgments}

The authors would like to thank Shir Peleg and Amir Shpilka for several helpful conversations throughout the course of this work.

\section{Preliminaries}

In this section we establish notation and preliminary facts we will need for the rest of the paper.
Let $S = \bC\bs{x_{1}, \ldots, x_{N}}$ denote the polynomial ring, graded by degree $S = \bigoplus_{i \geq 0} S_{i}$.
Given a vector space $V \subset S$, we use $V_{i}$ to denote the degree $i$ piece, that is, $V_{i} = V \cap S_{i}$.
We say that a vector space is graded if $V = \oplus V_{i}$.

We use \emph{form} to refer to a homogeneous polynomial.
Given two forms $A, B$ we say that $A, B$ are non-associate if $A \not \in \ideal{B}$ and $B \not \in \ideal{A}$.
If $A, B$ are of the same degree, this is equivalent to them being linearly independent.

\subsection{Algebraic geometry preliminaries}

We begin by recalling the notion of absolute irreducibility.

\begin{definition}
 A polynomial $P\in \bF[x_1,\cdots,x_N]$ is called \emph{absolutely irreducible} if $P$ is irreducible in the polynomial ring $\overline{\bF}[x_1,\cdots,x_n]$, over an algebraic closure $\overline{\bF}$ of $\bF$.
\end{definition}

Let $S = \bK[x_{1}, \dots, x_{n}, y_{1}, \dots, y_{m}]$ be a graded polynomial ring with $\deg(x_i)=d_i\geq 1$ and  $\deg(y_j)=e_j\geq 1$. 
Let $A=\bK[x_1,\cdots,x_n]$ and $K(A)$ the fraction field of $A$. 
Let $d\geq 1$ and $P \in S_d \setminus \ideal{x_{1}, \dots, x_{n}}$ such that $P$ is absolutely irreducible over $K(A)$.

The following two propositions bound the number of forms $Q_{i} \in A$ such that $\ideal{P, Q_{i}}$ is not prime, and the number of such forms such that $P + Q_{i}$ is absolutely reducible respectively.
The first of these was proved in \cite{GOS24}.
The novelty in the below statements is the fact that the bounds are independent of $n, m$.

\begin{theorem}[{\cite[Theorem~4.16]{GOS24}}]
\label{lem:primebound}
There are at most $\primebound{d}$ pairwise non-associate irreducible homogeneous elements $Q_{i} \in A$ such that $\ideal{P, Q_{i}}$ is not prime. 
\end{theorem}

\begin{proposition}
\label{lem:absredcount}
There are at most $\irredprimebound{d}$ elements $Q_{i} \in A$ such that $P + Q_{i}$ is absolutely reducible over $K(A)$.
\end{proposition}

\begin{proof}
    The proof proceeds in two parts.
    We first show the existence of a function $\widetilde{\cC}_{2}: \bN \times \bN \to \bN$ such that the number of $Q_{i}$ is at most $\widetilde{\cC}_{2}(m, d)$.
    Let $\bF := \overline{K(A)}$, the algebraic closure of $K(A)$.
    Let $S' := \bF\bs{y_{1}, \dots, y_{m}}$.
    Let $M := \binom{m+d+1}{d}$, so $\bP_{S'}^{M-1}$ can be identified as the space of $m$-variate forms of degree $d$ with coefficients in $S'$.
    By \cite[Corollary~14.3]{Eis95}, the set of reducible forms are a closed subvariety of $\bP_{S'}^{N-1}$, say $V_{m, d}$.
    Moreover, the degree of $V_{m, d}$ depends only on $m, d$.
    The ideal $I(V_{m, d})$ can be generated by forms whose degree only depends on $m, d$, let $\widetilde{C}_{2}(m, n)$ be the maximum degree of a generating set for $I(V_{m, d})$.

    By assumption, $P$ is irreducible in $S'$, therefore $P \not \in V_{m, d}$.
    Let $F$ be a form in the ideal of $V_{m, d}$ of degree at most $\widetilde{\cC}_{2}(m, d)$ such that $F(P) \neq 0$.
    Let $F'$ be the univariate obtained by specialising all the variables of $F$ to the coefficients of $P$, except the constant term.
    Let $P_{0}$ be the constant term of $P$.
    We know that $F'$ is nonzero, since $F'(P_{0}) = F(P) \neq 0$.
    Every $Q_{i} \in \bF$ such that $P + H_{i}$ is absolutely reducible must satisfy $F'(P_{0} + Q_{i}) = 0$.
    Since $\deg F' \leq \widetilde{\cC}_{2}(m, d)$, there are at most $\widetilde{\cC}_{2}(m, d)$ many $Q_{i}$.

    We now claim that $\irredprimebound{d} = \max_{c \leq d, 1 \leq m' \leq 5} \widetilde{\cC}_{2}(m', c)$ satisfies the required properties.
    If $m \leq 5$, then the claim is true by definition of $\cC_{2}$.
    Assume $m \geq 6$, and towards a contradiction assume that there are more than $\irredprimebound{d}$ forms such that $P + Q_{i}$ is absolutely reducible.
    Apply $m-5$ random hyperplane sections to $P$, to obtain a degree $c$ polynomial $P'$.
    By \cite[Lemma~4.7]{OS24}, the form $P'$ is irreducible in $S'$.
    Since the hyperplane sections do not act on $\bF$, each of the forms $P' + Q_{i}$ is absolutely reducible, which contradicts the definition of $\widetilde{\cC}_{2}(5, c)$.
\end{proof}

\begin{corollary}
    \label{cor:absredspan}
    There are at most $\irredprimebound{d}$ non-associate homogeneous degree $d$ forms $Q_{i} \in A$ such that $\Kspan{P, Q_{i}} \setminus \bc{Q_{i}}$ contains a form that is absolutely reducible over $K(A)$.
\end{corollary}
\begin{proof}
    Suppose $\lambda_{i} \in \bK$ is such that $P + \lambda_{i} Q_{i}$ is absolutely reducible.
    Since the forms $Q_{i}$ are non-associate, we have $\lambda_{i} Q_{i} \neq \lambda_{j} Q_{j}$ for any $i \neq j$.
    Applying \cref{lem:absredcount} completes the proof.
\end{proof}

\section{Sylvester-Gallai \& Edelstein-Kelly Configurations}\label{section: SG and EK configurations}

In this section, we formally define Sylvester–Gallai (SG) configurations, and Edelstein–Kelly (EK) configurations.
We start by defining linear versions of these configurations, and stating known bounds on them.
We then define higher degree EK configurations, which are the main object of study in this article.

\subsection{Linear Sylvester–Gallai configurations}

We start with the basic definition of robust linear Sylvester–Gallai (SG) configurations.

\begin{definition}[Robust linear Sylvester-Gallai configurations]
\label{def:linearsg}
Let $r \in \bN$ , $0 < \delta \leq 1$ and $V$ be a $\bK$-vector space. Let $\cF := \{v_1, \ldots, v_m\} \subset V$ be a finite set of \emph{pairwise linearly independent} vectors. 
We say that $\cF$ is a $\lrsg{r}{\delta}$ configuration over $\bK$ if there exists a $\bK$-vector subspace $U \subset V$ of dimension at most $r$ such that the following condition holds: 
\begin{itemize}
    \item for any $v_i \in \cF \setminus U$, there exist at least $\delta (m-1)$ indices $j \in [m] \setminus \{i\}$ such that $v_j \not\in U$ and
    $$\abs{ \Kspan{v_i, v_j} \cap \cF } \geq 3 \ \ \ \text{ or } \ \ \  \Kspan{v_i, v_j} \cap U \neq (0).$$
\end{itemize} 

We will say that $\cF$ is a $\lrsg{r}{\delta}$ configuration over the vector space $U$.
\end{definition}

The following bound on such configurations is proved in \cite[Proposition 3.5]{OS24}, which is a generalization of the result \cite[Corollary 16]{S20}, using the sharper bounds from \cite{DGOS18}.

\begin{proposition}
    \label{prop:linearsg}
    If $\cF$ is a $\lrsg{r}{\delta}$ then $\dim \Kspan{\cF} \leq r + 12 / \delta$.
\end{proposition}

\subsection{Edelstein–Kelly configurations}

We now define partial linear Edelstein–Kelly (EK) configurations.
These generalise the classical notion of EK configurations in three ways.
The first is that some forms are allowed to not satisfy any relationship.
The second is that a form in a given set is only required to have span dependencies with forms in the bigger of the remaining two sets.
The third is that we only require dependencies with a constant fraction of this bigger set.

\begin{definition}[Partial linear Edelstein–Kelly configurations]
\label{def:linearek}
Let $c \in \bN$ , $0 < \delta \leq 1$ and $V$ be a $\bK$-vector space.
Let $\cF := \{v_1, \ldots, v_m\} \subset V$ be a finite set of pairwise linearly independent vectors.
We say that $\cF$ is a $\lrek{r}{\delta}$ configuration over $\bK$ if there exists a subset $\cG \subset \cF$ of size at most $r$, and a partition $\cF_{1} \cup \cF_{2} \cup \cF_{3}$ of the remaining forms $\cF \setminus \cG$ (with $\abs{\cF_{1}} \geq \abs{\cF_{2}} \geq \abs{\cF_{3}}$) such that the following condition holds:
\begin{itemize}
    \item for any $v \in \cF_{1}$ there are at least $\delta \abs{\cF_{2}}$ forms $u \in \cF_{2}$ such that $\abs{\Kspan{v, u} \cap \br{\cF_{3} \cup \cG}} \geq 1$. 
    \item for any $v \in \cF_{2}$ there are at least $\delta \abs{\cF_{1}}$ forms $u \in \cF_{1}$ such that $\abs{\Kspan{v, u} \cap \br{\cF_{3} \cup \cG}} \geq 1$. 
    \item for any $v \in \cF_{3}$ there are at least $\delta \abs{\cF_{1}}$ forms $u \in \cF_{1}$ such that $\abs{\Kspan{v, u} \cap \br{\cF_{2} \cup \cG}} \geq 1$. 
\end{itemize} 
\end{definition}

Such configurations were introduced by \cite{S20} and further studied by \cite{PS20b}.
The following is an improved bound on the dimensions of such configurations, and follows by a more careful analysis of the proof of \cite{S20}, and the above improved bound on the dimension of $\lrsg{r}{\delta}$ configurations.
For completeness we provide a proof here.

\begin{proposition}
    \label{prop:linearek}
    If $\cF$ is a $\lrek{r}{\delta}$ configuration then $\dim \Kspan{\cF} \leq r + c_{\mathrm{ek}} \cdot \delta^{-1} \log \br{\delta^{-1}}$ for a universal constant $c_{\mathrm{ek}}$.
\end{proposition}

\begin{proof}
    Let $\cF_{1}, \cF_{2}, \cF_{3}, \cG$ be a partition of $\cF$ that satisfies the properties of the definition.
    Let $\delta' := \delta / 2$.
    The main step is to find a subset $\cF'_{1} \subset \cF_{1}$ such that $\abs{\cF_{2}} \leq \abs{\cF'_{1}} \leq 2 \abs{\cF_{2}}$ and such that $\cF'_{1} \cup \cF_{2} \cup \cF_{3} \cup \cG$ is itself a $\lrek{r}{\delta'/2}$ configuration, with partition $\cF'_{1}, \cF_{2}, \cF_{3}$.
    If such a set $\cF'_{1}$ exists, then the above configuration is actually an $\lrsg{r}{\delta'/8}$, since each form satisfies as span relation with at least $\delta'/4$ fraction of the total forms.
    Therefore by \cref{prop:linearsg} we can deduce that $\dim \Kspan{\cF'_{1} \cup \cF_{2} \cup \cF_{3} \cup \cG} \leq r + 96 \cdot \delta'^{-1}$.
    Every form $F \in \cF_{1} \setminus \cF'_{1}$ is contained in $\Kspan{\cF_{2} \cup \cF_{3} \cup \cG}$, since this form satisfies a linear condition with at least one form in $\cF_{2}$.
    Therefore we can also deduce $\dim \Kspan{\cF} \leq r + 96 \cdot \delta'^{-1} = r + 192 \cdot \delta^{-1}$, which will complete the proof.

    It suffices therefore to show that such a set $\cF'_{1}$ exists.
    We show that it exists as long as $\abs{\cF_{2}}$ is bigger than $c \cdot \delta^{-1} \log \br{\delta^{-1}}$ for some fixed constant $c$.
    If $\abs{\cF_{2}}$ is smaller than $c \cdot \delta^{-1} \log \br{\delta^{-1}}$ then so is $\abs{\cF_{3}}$.
    In this case we have $\dim \Kspan{\cF_{2} \cup \cF_{3} \cup \cG} \leq r + 2 c \cdot \delta^{-1} \log \br{\delta^{-1}}$, and therefore $\dim \Kspan{\cF} \leq 2 + c \cdot \delta^{-1} \log \br{\delta^{-1}}$, which will also complete the proof.
    We show $\cF'_{1}$ exists by a standard probabilistic method argument.

    Let $m_{i} := \abs{\cF_{i}}$ and let $p := m_{2} / m_{1}$.
    If $m_{3} \leq 2 m_{2}$ then $\cF'_{1} = \cF_{1}$ satisfies the required properties and we are done, so we can assume $p \leq 1/2$.
    Let $\cF'_{1} \subset \cF_{1}$ be a subset constructed by picking each element of $\cF_{1}$ independently with probability $p$.

    For each $v \in \cF_{2}$, let $\Fspan[v] \subset \cF_{1}$ be the set of forms $u$ such that $\abs{\Kspan{v, u} \cap \br{\cF_{3} \cup \cG}} \geq 1$.
    We have $\abs{\Fspan[v]} \geq \delta m_{1}$.
    The size of the subset $\cF'_{1} \cap \Fspan[v]$ is a random variable with expectation $\abs{\Fspan[v]} m_{2} / m_{1} \geq \delta m_{2}$, and by a Chernoff bound, the probability that it has size less than $\delta' m_{2}$ is at most $2^{-a \delta m_{2}}$ for a fixed constant $a$.
    Similarly, for $v \in \cF_{3}$ we can also define $\Fspan[v]$, and the probability that $\Fspan[v]$ has size less than $\delta ' m_{2}$ is bounded by $2^{-a \delta m_{2}}$.
    Finally, we note that $\abs{\cF'_{1}}$ is itself a random variable with expectation $m_{2}$, and the probability that $\abs{\cF'_{1}} \geq 2m_{2}$ is bounded by $2^{-bm_{2}}$.
    By a union bound, as long as $\br{m_{2} + m_{3}} 2^{-a \delta m_{2}} + 2^{-bm_{2}} < 1$, there is a set $\cF'_{1}$ such that $\abs{\cF'_{1}} \leq 2m_{2}$ and such that $\abs{\Fspan[v] \cap \cF'_{1}} \geq \delta' m_{2}$ for every $v \in \cF_{2} \cup \cF_{3}$.
    Since $m_{3} \leq m_{2}$, the above condition holds as long as $m_{2} \geq c \cdot \delta^{-1} \log \br{\delta^{-1}}$ for a big enough constant $c$ (this constant only depends on $a, b$).

    Fix a set $\cF'_{1}$ that satisfies the above conditions.
    If $\abs{\cF'_{1}} < m_{2}$, then we arbitrarily add elements to $\cF'_{1}$ to ensure that $\abs{\cF'_{1}} = m_{2}$, this does not break any of the conditions.
    These conditions exactly guarantee that $\cF'_{1} \cup \cF_{2} \cup \cF_{3} \cup \cG$ is a $\lrek{r}{\delta'/2}$ configuration.
    This completes the proof.
\end{proof}

We now define higher degree Edelstein–Kelly configurations, which are our main object of study.

\begin{definition}[High degree EK configurations]
\label{def:ek}
Let $U \subset S$ be a graded finitely generated vector space such that $R := S / \ideal{U}$ is a UFD, and let $z \in R_{1}$. 
Let $\cA, \cB, \cC \subset R$ be pairwise disjoint finite sets of irreducible forms of degree at most $d$.
We say that $\cA, \cB, \cC$ is a $(d, z, R)$-EK configuration if the following conditions hold: 
\begin{enumerate}
    \item $z \not\in \cA \cup \cB \cup \cC$, and the union $\{z\}\cup \cA \cup \cB \cup \cC$ consists of pairwise non-associate forms.
    \item For every $A \in \cA, B \in \cB$ we have
    $$ z \cdot \prod_{C \in \cC} C \in \radideal{A, B}.$$
    \item For every $A \in \cA, C \in \cC$ we have
    $$ z \cdot \prod_{B \in \cB} B \in \radideal{A, C}.$$
    \item For every $B \in \cB, C \in \cC$ we have
    $$ z \cdot \prod_{A \in \cA} A \in \radideal{B, C}.$$
\end{enumerate}
\end{definition}

In the case when the ideals occurring are prime, for example if $\ideal{A, B}$ is prime for some $A \in \cA, B \in \cB$, then the condition above implies that either $z \in \ideal{A, B}$ or $C \in \ideal{A, B}$ for some $C \in \cC$.
We then say that $C$ is the EK-image of $A, B$.
The image $C$ might not be unique, we arbitrarily fix a choice for every pair.

\section{Strong Algebras}\label{sec:strongalg}

In this section we recall the definitions and results that we will need about strong algebras.
Most of the definitions and results from this section are taken from \cite[Section 5]{OS24}, and we refer the reader to this paper for motivation and further discussions about these definitions.

\subsection{Strength}

Let $R=\oplus_{d\geq 0} R_d$ be a finitely generated graded $\bK$-algebra,  generated by $R_1$. 
In \cite{AH20} the notions of collapse and strength were defined for a polynomial ring. 
We will extend those definitions to finitely generated graded $\bK$-algebras and prove the necessary properties below.
Henceforth, we refer to a homogeneous element of $R$ as a \emph{form}, adopting the same notation for polynomial rings.

\begin{definition}[Collapse]\label{def:collapse} 
Given a non-zero form $F \in R_d$, we say that $F$ has a \emph{$k$-collapse} if there exist $k$ forms $G_1, \ldots, G_k$ such that $1 \leq \deg(G_i) < d$ and $F \in (G_1, \ldots, G_k)$.
\end{definition}

\begin{definition}[Strength]\label{def:strength} 
Given a non-zero form $F \in R_d$, the \emph{strength} of $F$, denoted by $s(F)$, is the least positive integer such that $F$ has a $(s(F)+1)$-collapse but it has no $s(F)$-collapse.
We say that $s(F) \geq t$ whenever $F$ does not have a $t$-collapse.
\end{definition}

\begin{remark}
By the definitions above, a form $x\in R_1$ does not have a $k$-collapse for any $k \in \bN$. 
Thus, we say that for any $x \in R_1$, $s(x) = \infty$. In particular, linear forms in the polynomial ring $S$ have infinite strength.
We will make the convention that $s(0) = -1$.
\end{remark}

\begin{definition}[Minimum collapse]\label{definition: minimum collapse}
Given a non-zero form $F \in R_d$ and $s \in \bN^*$ such that $s(F) = s-1$, a \emph{minimum collapse} of $F$ is any identity of the form $F = G_1 H_1 + \cdots + G_s H_s$, where $G_i, H_i$ are forms of degree in $[d-1]$.
\end{definition}

It is useful to define the min and max strength of a linear system of forms of the same degree.

\begin{definition}[Min and max strength]\label{definition: smin smax}
Given a set of forms $F_1, \ldots, F_r \in R_d$,  define 
$\smin{F_1, \ldots, F_r}$ as the \emph{minimum} strength of a \emph{non-zero} form in $\Kspan{F_1, \ldots, F_r}$ and $\smax{F_1, \ldots, F_r}$ as the \emph{maximum} strength of a form in $\Kspan{F_1, \ldots, F_r}$. 

In particular, given any non-zero finite dimensional vector space $V \subset R_d$, define $\smin{V}$ ($\smax{V}$) as the minimum (maximum) strength of any non-zero form in $V$. If $V=(0)$, then there are no non-zero forms in $V$. In this case, by convention we define $\smin{(0)}=\smax{(0)}=\infty$.
We will say that a vector space $V$ is $k$-strong if $\smin{V}\geq k$. Note that the zero vector space is infinitely strong.
\end{definition}

\subsection{Strong Ananyan-Hochster Vector Spaces}\label{section: AH}

Let $R=\oplus_{d\geq 0} R_d$ be a finitely generated graded $\bK$-algebra, generated by $R_1$. 
Given a graded $\bK$-vector space $V = \bigoplus_{i=1}^d V_i \subset R$, where $\delta_i := \dim V_i$, we denote its dimension sequence by $\delta := (\delta_1, \ldots, \delta_d)$.

\begin{definition}[Strong Ananyan-Hochster vector spaces]\label{def:strong-AH-algebra} 
Let $R=\oplus_{d\geq 0} R_d$ be a finitely generated graded $\bK$-algebra, generated by $R_1$. 
For any function $B=(B_1,\cdots,B_d):\bN^d\rightarrow \bN^d$, we say that a non-zero graded vector subspace $V=\oplus_{i=1}^dV_i\subset R$, with dimension sequence $\delta$, is a $B$-strong AH vector space if $V_i$ is $B_i(\delta)$-strong for all $i$, i.e. $\smin{V_i}\geq B_i(\delta)$.  
The subalgebra $\bK[V]\subset R$ generated by a $B$-strong AH vector space $V$ is called a $B$-strong AH algebra.
\end{definition}

Note that if $V=(0)$, then $V$ is $B$-strong for any function $B$, since $\smin{(0)}=\infty$.
The following result is a corollary of \cite[Theorem A]{AH20}, and a proof can be found in \cite[Corollary 5.9]{OS24}.
In the following lemma and the rest of this article, the function $A(\eta, d): \bN^{2} \to \bN$ is the function defined in \cite[Theorem A]{AH20}.

\begin{corollary}\label{corollary: strong sequence R_eta}
Let $V=\oplus_{i=1}^d V_i\subset S$ be a $B$-strong AH vector space for some $B:\bN^d\rightarrow \bN^d$.
Suppose $B_i(\delta)\geq A(\eta,i)+3(\sum_i\delta_i-1)$ for some $\eta\in \bN$. Then any sequence of $\bK$-linearly independent forms in $V$ is an $\mathcal{R}_\eta$-sequence. If $\eta\geq 3$, then $S/(V)$ is a Cohen-Macaulay, unique factorization domain.
\end{corollary}

\subsection{Lifted strength}

\begin{definition}[Lifted strength] 
Let $U\subset S$ be a graded vector space and $R=S/(U)$. 
Let $F\in R_d$ be a non-zero form. 
We define the lifted strength of $F$ with respect to $U$ as $$\lsmin{U,F}:=\min\{\smin{U_d+\Kspan{\widetilde{F}}}\}$$
where $\widetilde{F}$ varies over all forms in $S_d$ such that the image of $\widetilde{F}$ in $R$ is $F$. 
Given a set of forms $F_1,\cdots,F_m\in R_d$, we define $$\lsmin{U,F_1,\cdots,F_m}=\min\{\smin{U_d+\Kspan{\widetilde{F}_1,\cdots,\widetilde{F}_m}}\},$$ where $\widetilde{F}_i$ varies over all forms in $S_d$ such that the image of $\widetilde{F}_i$ in $R$ is $F_i$. 
Given a non-zero vector space $V\subset R_d$, we define $$\lsmin{U,V}=\min\{\lsmin{U_d,F_1,\cdots,F_m}\},$$ where $F_1,\cdots,F_m$ vary over all possible bases of $V$. 
We say that $V\subset R_d$ is $k$-lifted strong with respect to $U$ if $\lsmin{U,V}\geq k$. 
For simplicity, we omit $U$ from the notation and write $\lsmin{V}$ when $U$ is clear from the context.

Suppose that $U\subseteq S_{\leq d}$ is of dimension sequence $\delta_U$. Let $V=\oplus_{i=1}^dV_i\subset R$ be a graded vector space of dimension sequence $\delta_V$.  
For any function, $B:\bN^d\rightarrow \bN^d$ we will say that $V$ is $B$-lifted strong with respect to $U$, if $V_i$ is $B_i(\delta_U+\delta_V)$-lifted strong, i.e. $\lsmin{U,V_i}\geq B_i(\delta_U+\delta_V)$ for all $i\in [d]$.  
In other words, $V$ is $B$-lifted strong with respect to $U$, if the vector space $U+\Kspan{\widetilde{F}_1,\cdots,\widetilde{F}_m}$ is $B$-strong in $S$, for any homogeneous basis $F_1,\cdots,F_m\in R$ of $V$ and any set of homogeneous lifts $\widetilde{F}_1,\cdots,\widetilde{F}_m\in S$.
\end{definition}

\subsection{Strengthening and Robustness}

For any $\mu\in \bN^d$, we define the translation function $t_\mu:\bN^d\rightarrow\bN^d$ as $t_\mu=(t_{\mu,1},\cdots, t_{\mu,d})$ where the $i$-th component is defined by $t_{\mu,i}(\delta)=\delta_i+\mu_i$. In other words, for all $i\in [d]$ we add $\mu_i$ to the $i$-th component of $\delta$. For any $n\in \bN$, we let $t_n:=t_{(n,\cdots,n)}$.

The following lemma is proved in \cite[Lemma 5.15]{OS24}.

\begin{lemma}[Strengthening of Algebras]\label{proposition: constructing AH algebras} For any $d\in \bN$ and a function $B:\bN^d\rightarrow N^d$,
   there exist functions  $C_B:\bN^d\rightarrow\bN^d$ and $h_B:\bN^d\rightarrow\bN^d$, depending on $B$, such that the following holds:
    
    Given a graded vector space  $U=\oplus_{i=1}^dU_i\subset S$ with dimension sequence $\delta\in \bN^d$, there exists a $B$-strong AH vector space $V=\oplus_{i=1}^d V_i$ such that 
    \begin{enumerate}
        \item $\bK[U]\subset \bK[V]$,
        \item for all $i\in [d]$, we have $\dim(V_i)\leq C_{B,i}(\delta)$, where $C_{B,i}$ denotes the $i$-th component of $C_B=(C_{B,1},\cdots, C_{B,d}):\bN^d\rightarrow \bN^d$.
    \end{enumerate}
    
Furthermore, suppose   $H=\oplus_{i=1}^dH_i\subset U$ is a graded subspace such that $\smin{H_i}\geq h_{B,i}(\delta)$ for all $i\in [d]$. 
Then there exists a $B$-strong AH vector space $V$  satisfying (1) and (2) above such that $H\subset V$.
\end{lemma}

The following corollary is from \cite[Corollary 5.16]{OS24}.

\begin{corollary}[Robustness of strong algebras]\label{corollary: robust strong algebra} Let $B,G:\bN^d\rightarrow\bN^d$ and $\mu\in\bN^d$. Suppose that $B_i(\delta)\geq h_{G,i}(\delta+\mu)$ for all $\delta\in \bN^d$ and $i\in [d]$, where $h_G:\bN^d\rightarrow \bN^d$ is the function defined in \cref{proposition: constructing AH algebras}. Let $U\subset S$ be a $B$-strong AH vector space and $W\subset S$ is a graded vector space with dimension sequences $\delta$ and $\mu$ respectively. Then there exists a $G$-strong AH vector space $V$ such that 
\begin{enumerate}
    \item $\bK[U+W]\subset\bK[V]$,
    \item $U\subset V$,
    \item for all $i\in [d]$, $\dim(V_i)\leq C_{G,i}(\delta+\mu)$, where $C_G:\bN^d\rightarrow \bN^d$ is the function defined in \cref{proposition: constructing AH algebras}.
\end{enumerate}
\end{corollary}

The following corollary corresponds to \cite[Corollary 5.17]{OS24}.

\begin{corollary}\label{corollary: practical robustness}
 Let $B:\bN^d\rightarrow \bN^d$. Let $U\subset S$ be a graded vector space with dimension sequence $\delta_U\in \bN^d$ and let $R=S/(U)$. Let $V\subset R$ is a graded vector space with dimension sequence $\delta_V\in \bN^d$. Suppose $V$ is  $h_{2B}\circ t_k$-lifted strong with respect to $U$. Let $P_1,\cdots,P_k\in R_{\leq d}$ be homogeneous elements. Then there exists a graded vector space $V'\subset R_{\leq d}$ such that:

\begin{enumerate}
    \item $V'$ is $B$-lifted strong with respect to $U$.
    \item $P_1,\cdots,P_k\in \bK[V']$.
    \item $V\subset V'$.
    \item for all $i\in [d]$, we have $\dim(V'_i)\leq C_{2B,i}(t_k(\delta_U+\delta_V))-\delta_{U,i}$. In particular, $\dim(V')\leq C(B,\delta_U,\delta_V,k):=\sum(C_{2B,i}(t_k(\delta_U+\delta_V))-\delta_{U,i})$.
\end{enumerate}
\end{corollary}

\begin{definition}[Ananyan-Hochster spaces]\label{def: AH-process}
    In the situation of \cref{corollary: robust strong algebra}, we define $AH(U,W)$ to be any graded vector space $V$ provided by \cref{corollary: robust strong algebra}.
    Similarly, in the situation of \cref{corollary: practical robustness}, we define $AH_R(V,P_1,\cdots,P_k)$ to be any vector space $V'$ provided by \cref{corollary: practical robustness}.
\end{definition}

\begin{remark}
    Note that, for given $V,P_1,\cdots,P_k$, the vector space $AH_R(V,P_1,\cdots,P_k)$ is not necessarily unique. 
    As stated in our definition, we use the notation $AH_R(V,P_1,\cdots,P_k)$ to denote any vector space satisfies the properties in \cref{corollary: practical robustness}, whose existence is guaranteed. 
    We will only use these properties of these spaces and in all our arguments we work with any fixed choice of such a vector space $AH_R(V,P_1,\cdots,P_k)$. 
\end{remark}

\paragraph{An iterative AH-process.} In order to construct strong vector spaces of uniformly bounded dimension, we will iteratively apply the AH-construction from \cref{def: AH-process}. 
We make a convenient definition for such iterative processes.

\begin{definition}[$(k,t)$-process]
    Let $U \subseteq S$ be a graded finitely generated vector space such that $R := S / \ideal{U}$ is a UFD, and $k, t \geq 1$ be integers.
    Let $V \subseteq R$ be a vector space and let $\cF \subseteq R$.
    A $(k, t)$ process on $V$ is defined to be a process that starts with the vector space $V$, and performs at most $t$ rounds.
    In each round, at most $k$ forms $F_{1}, \dots, F_{k}$ of degree at most $d$ are picked from the set $\cF$, and $V$ is updated to $AH_R(V, F_{1}, \dots, F_{k})$.
    In each iteration, the picked forms are allowed to be chosen based on the new $V$.
    Note that we are dropping the dependence on $d$ from the notation for convenience, as $d$ will always be clear from context.
\end{definition}

\paragraph{Strength and dimension bound functions.} 
We introduce some auxiliary functions that will provide uniform bounds for $(k,t)$-processes (see  \cref{lem: iterated lifted strength}). 

\emph{Strength bound function.} Given a function $B: \bN^d\rightarrow \bN^d$, and positive integers $k, t$, we define a function $H(B, k, t): \bN^d\rightarrow \bN^d$ recursively as follows.
Let $H(B, k, 0) := B$, and \[H(B, k, t) := h_{2 H(B, k, t-1)} \circ t_{k + 1}, \] where $h$ is as defined in \cref{proposition: constructing AH algebras}.
This function $H(B,k,t)$ captures the strength needed for a vector space $V$, so that after applying a $(k,t)$-process to $V$, the resulting vector space $V'$ is still $B$-lifted strong and this $(k,t)$-process preserves $V$, i.e. $V\subseteq V'$.

\emph{Dimension bound function.} Fix $\delta \in \bN^{d}$, and $B,k,t$ as above. We define a function $D(B, k, t, \delta) : \bN^{d} \to \bN^{d}$ recursively as follows.
Let $D(B, k, 0, \delta)$ be the identity function.
To define $D(B, k, t, \delta)(\mu)$, assuming  $D(B, k, t-1, \delta)$ is already defined, we first define \[\eta := \max_{\substack{x\in \bN^d\\ \norm{x}_{1} = \norm{\mu}_{1}}} \norm{D(B, k, t-1, \delta)(x)}_{1},\]
where $\norm{x}_1=x_1+\cdots+x_d$ for $x=(x_1,\cdots,x_d)\in \bN^d$.
We then define the $i^{th}$ coordinate function $D_{i}(B, k, t, \delta)(\mu)$ as
$$ D_{i}(B, k, t, \delta)(\mu) = \max_{x, \norm{x}_{1} = \eta} C_{2 H(B, k, t), i}(t_{k+1} (x + \delta)),$$
where again $C$ is the function defined in \cref{proposition: constructing AH algebras}.
This function $D(B,k,t,\delta)$ captures the dimension upper bound for the result of a $(k,t)$-process applied on a vector space $V$. 

We note these bounds and basic properties of $(k,t)$ processes below. 
\begin{lemma}
    \label{lem: iterated lifted strength}
     Let $U \subseteq S$ be a graded finitely generated vector space such that $R := S / \ideal{U}$ is a UFD, and $k, t, k', t' \geq 1$ be integers.
     \begin{enumerate}
        \item The composition of a $(k, t)$ process and a $(k', t')$ process is a $(\max\br{k, k'}, r + r')$ process.
        \item If $V$ is $H(B, k, t)$-lifted strong, and if $V'$ is the result of a $(k, t)$ process on $V$, then $V \subset V'$ and $V'$ is $B$-lifted strong.
        Moreover, each intermediate vector space that appears as part of the process is $h_{2B} \circ t_{1}$-lifted strong.
        \item If $V'$ is the result of a $(k, t)$ process on $V$, then the dimension of $V'$ is bounded by $D(B, k, t, \delta_{U})(\delta_{V})$, where $\delta_{U}, \delta_{V}$ are the dimension vectors of $V, U$ respectively.
\end{enumerate}
\end{lemma}
\begin{proof}
    The first property follows by definition.
    Properties $(2)$ and $(3)$ follow by induction.
\end{proof}

\subsection{Absolute irreducibility with respect to strong vector spaces}
\label{subsec:absirredstrong}

The following definition is from \cite{GOS24}.

\begin{definition}
    \label{def:absred}
    Let $B:\bN^d\rightarrow \bN^d$.
    Suppose $B_i(\delta)\geq A(\eta,i)+3(\sum_i\delta_i-1)$ for some $\eta\in \bN$.
    Suppose $R = S / \ideal{U}$.
    Suppose $V \subset R$ is a graded vector space that is $h_{2B} \circ t_{1}$-lifted strong with respect to $U$.
    Suppose $P \in R$ is a form.
    Let $V'$ be the vector space obtained by applying \cref{corollary: practical robustness} to $V$ and $P$.
    Let $y_{1}, \dots, y_{a}$ be a basis of homogeneous forms of $V$, and $y_{a+1}, \dots, y_{b}$ extend this to a basis of $V'$.
    We say $P$ is absolutely reducible over $V$ if $P$ is absolutely reducible as a polynomial in $\bK\br{y_{1}, \dots, y_{a}}\bs{y_{a+1}, \dots, y_{b}}$.
\end{definition}

Absolute reducibility and irreducibility with respect to vector spaces was defined in \cite{GOS24} to allow \cref{lem:primebound} to be applied in more general settings.
In the same setting, we can also apply \cref{cor:absredspan}.
The functions $\primebound{d}, \irredprimebound{d} : \bN \to \bN$ referred to in the following lemma is the same functions whose existence is guaranteed by \cref{lem:primebound}, \cref{cor:absredspan} respectively.

\begin{corollary}[Corollary of \cref{lem:primebound}, \cref{cor:absredspan}]
    \label{cor:strongprimebound}
    Let $B:\bN^d\rightarrow \bN^d$.
    Suppose $B_i(\delta)\geq A(\eta,i)+3(\sum_i\delta_i-1)$ for some $\eta\in \bN$.
    Suppose $R = S / \ideal{U}$.
    Suppose $V \subset R$ is a graded vector space that is $h_{2B} \circ t_{1}$-lifted strong with respect to $U$.
    Suppose $P \in R$ is absolutely irreducible with respect to $V$ such that $\deg P = d$.
    There are at most $\primebound{d}$ irreducible non-associate forms $Q_{i} \in \bK\bs{V}$ such that $\ideal{Q_{i}, P}$ is not prime.
    Further, there are at most $\irredprimebound{d}$ pairwise non-associate degree $d$ irreducible forms $H_{i} \in \bK\bs{V}$ such that $\Kspan{P, H_{i}} \setminus \bc{H_{i}}$ contains a form that is absolutely reducible over $V$.
\end{corollary}
\begin{proof}
    Let $V'$ be the vector space obtained by applying \cref{corollary: practical robustness} to $V$ and $P$.
    Let $y_{1}, \dots, y_{a}$ be a basis of homogeneous forms of $V$, and $y_{a+1}, \dots, y_{b}$ extend this to a basis of $V'$.
    The ring $\bK\bs{V'}$ is isomorphic to a polynomial ring in the variables $y_{i}$.
    Further, $P$ is irreducible as a form in $\bK\bs{V'}$ by definition of absolute irreducibility with respect to $V$.
    
    Therefore we can apply \cref{lem:primebound} to deduce that there are at most $\primebound{d}$ forms $Q_{i} \in \bK\bs{V}$ such that $(P, Q_{i})$ is not prime as an ideal of $\bK\bs{V'}$.
    Further, since $y_{1}, \dots, y_{b}$ form a prime sequence, prime ideals of $\bK\bs{V'}$ extend to prime ideals in $R$.
    
    Similarly, by \cref{cor:absredspan}, there are at most $\irredprimebound{d}$ forms $H_{i} \in \bK\bs{V}$ such that $\ideal{P, H_{i}}_{d} \setminus \bc{H_{i}}$ contains a form that is absolutely reducible over $V$, when we treat $\ideal{P, H_{i}}$ as an ideal in $\bK\bs{V'}$.
    Since both $P, H_{i}$ have degree $d$, the ideal $\ideal{P, H_{i}}$ contains the same elements in degree $d$ when treated as an ideal of $\bK\bs{V'}$ and as an ideal of $R$.
    This completes the proof.
\end{proof}

\section{Graded Quotients}\label{sec:generalquot}

This section defines graded quotients and the essential properties needed for the degree reduction step in the proof of our main technical theorem.
The definitions and results in the first part of this section are from \cite{OS24, GOS24}.
The definitions and results in subsequent parts are new to this work.
For completeness, we state all the statements we need from previous works without proof, while we provide a proof of all the new/refined statements that we claim here.

Throughout this section, we fix positive integers $d,\eta\in \bN$ with $\eta\geq 3$.
Throughout this section, $B:\bN^d\rightarrow \bN^d$ denotes an ascending function such that $B_i(\delta)\geq A(\eta,i)+3(\sum_i\delta_i-1)$ for all $i\in [d]$. Here $A(\eta, i)$ is the function defined in \cref{section: AH}.

The following definition corresponds to \cite[Definition 6.1]{OS24}.

\begin{definition}[Graded Quotients]\label{definition: quotients}
Let $U=\oplus_{i=1}^d U_i\subset S$ be a graded vector space of dimension sequence $\delta$ in $S$ and $R := S/(U)$ be the quotient ring. 
Let $V=\oplus_{i=1}^dV_i \subset R$ be a graded subspace of dimension sequence $\mu$.

Let $F_1, \ldots, F_n$ be a homogeneous basis for $V$ and $z$ be a new variable. 
For $\alpha \in \bK^n$, let $I_\alpha\subseteq R[z]$ be the homogeneous ideal generated by the forms $ \{F_1 - \alpha_1 z^{\deg(F_1)}, \dots, F_n - \alpha_n z^{\deg(F_n)}\}$. 
We define the graded quotient map $\varphi_{V, \alpha}$ as the quotient homomorphism of finitely generated graded $\bK$-algebras given by
$$ \varphi_{V, \alpha} : R[z] \to R[z]/I_\alpha. $$
\end{definition}

\begin{remark}
The definition above depends on the choice of the basis $F_1,\cdots,F_n$ for $V$. We omit this from the notation for simplicity. We will apply this graded quotient construction for sufficiently strong vector spaces $V$. Therefore \cite{AH20} shows that any basis of $V_\alpha$ forms an $\cR_\eta$ sequence. Similarly, several of our statements regarding these graded quotients will apply independent of the choice of the basis. However, for certain applications we will deal with explicit choices of bases and we define the notion of compatible graded quotients in that context (\cref{def: compatible quotients}).
\end{remark}

\noindent The next proposition and lemma correspond to \cite[Proposition 6.3]{OS24} and \cite[Lemma 6.4]{OS24}.

\begin{proposition}\label{proposition: graded quotient is UFD}
Suppose $V\subset R$ is $B$-lifted strong with respect to $U$. Then $R[z]$ and $R[z]/I_\alpha$ are quotients of $S[z]$ by $\cR_\eta$-sequences, for any choice of $\alpha \in \bK^n$. 
In particular, they are Cohen-Macaulay UFDs.    
\end{proposition}
\emph{General points.} We say that a property $\mathcal{P}$ holds for a \emph{general} $ \alpha \in \bK^m$, if there exists a non-empty open subset $\cU \subseteq\bK^m$ such that the property $\mathcal{P}$ holds for all $\alpha\in \cU$. Here $\cU \subseteq \bK^m$ is open with respect to the Zariski topology. 
Hence $\cU$ is the complement of the zero set of finitely many polynomial functions on $\bK^m$. Note that, equivalently a property $\cP$ holds for a general $\alpha\in \bK^m$, if there is a closed subset $\cZ\subseteq \bK^m$ such that the $\cP$ holds for all $\alpha\not \in \cZ$.
\begin{lemma}\label{lemma: general projection basic}
Let $S=\bK[x_1,\cdots,x_N]$ and $z$ be a new variable. Fix positive integers $d_1,\cdots,d_n\in \bN$. For $\alpha\in \bK^n$, let $I_\alpha=(x_1-\alpha_1 z^{d_1},\cdots,x_n-\alpha_nz^{d_n})$. Let $\varphi_\alpha:S[z]\rightarrow S[z]/I_\alpha$ be the quotient ring homomorphism. 

\begin{enumerate}
    \item The ideal $I_\alpha$ is prime in $S[z]$, and the composition morphism $\bK[z]\hookrightarrow S[z]\rightarrow S[z]/I_\alpha$ is injective.
    \item  If $F\in \bK[x_1,\cdots,x_n]$ is a non-zero polynomial, then $\varphi_\alpha(F)\neq 0$ in $S[z]/I_\alpha$ for a general $\alpha\in \bK^n$.
    \item Let $F\in S\setminus \bK[x_1,\cdots,x_n]$, then $\varphi_\alpha(F)\not\in \bK[z]$ in $S[z]/I_\alpha$, for a general $\alpha\in \bK^n$.
    \item If $F\in S$ is a non-zero polynomial, then $\varphi_\alpha(F)\neq 0$ in $S[z]/I_\alpha$ for a general $\alpha\in \bK^n$.
    \item If $F,G\in S$ have no common factor, then $\gcd(\varphi_\alpha(F),\varphi_\alpha(G))\in\bK[z]$ for a general $\alpha\in \bK^n$. 
    \item If $F\in S$ is square-free. 
    Then, for a general $\alpha\in \bK^n$, the multiple factors of $\varphi_\alpha(F)$ must be in $\bK[z]$.
\end{enumerate}
\end{lemma} 

The next proposition corresponds to \cite[Proposition~6.5]{GOS24}.

\begin{proposition}\label{proposition: general quotient general}
Let $V \subset R$ be a $B$-lifted strong vector space and $\varphi_\alpha :R[z]\rightarrow R[z]/I_\alpha$ be a graded quotient as defined in \cref{definition: quotients}.
\begin{enumerate}
    \item The ideal $I_\alpha$ is a prime ideal in $R[z]$ and the composition $\bK[z]\hookrightarrow R[z]\rightarrow R[z]/I_\alpha$ is injective.
    \item  If  $F \in R \setminus \{0\}$, then $\varphi_\alpha(F)\neq 0$ for a general $\alpha\in \bK^n$.
    \item  If $F\not\in \bK[V] \subset R$, then $\varphi_\alpha(F)\not\in \bK[z]$ for a general $\alpha\in \bK^n$.
    \item If $F\not \in (V)$ then $\varphi_\alpha(F)\not \in (z)$ in $R[z]/I_\alpha$.
    \item If $F$ is absolutely reducible with respect to $V$ then $\varphi_{\alpha}(F)$ is a reducible form in $R[z]/I_\alpha$ for general $\alpha$.
    \item If $F$ is absolutely irreducible with respect to $V$ and $F\not\in (V)$, then $\varphi_\alpha(F)$ is irreducible for general $\alpha$.
\end{enumerate}
\end{proposition}

The next proposition corresponds to \cite[Proposition 6.6]{OS24}.

\begin{proposition}\label{proposition: gcd after projection}
Let $G:\bN^d\rightarrow \bN^d$ be a function such that $G_i(\delta)\geq h_{B,i}\circ t_2(\delta)$ for all $\delta\in \bN^d$. 
Let $V \subset R_{\leq d}$ be a $G$-lifted strong vector space and $\varphi_\alpha :R[z]\rightarrow R[z]/I_\alpha$ be a graded quotient as defined in \cref{definition: quotients}. 
Let $F, G \in R_{\leq d}$ be such that they have no common factor.
There exists a non-empty open subset $\cU\subseteq \bK^{\dim(V)}$ such that for all $\alpha\in \cU$ we have:
\begin{enumerate}
    \item $\gcd(\varphi_\alpha(F), \varphi_\alpha(G)) \in \bK[z]$. %
    \item If $F,G$ are homogeneous, then 
    $\gcd(\varphi_\alpha(F), \varphi_\alpha(G)) = z^k $ for some $k \in \bN$. 
    In particular, we have $\gcd(\varphi_\alpha(zF), \varphi_\alpha(zG)) = z^{k+1} $ for some $k \in \bN$. Furthermore, if $F,G\not \in K[V]\subset R$ then $\varphi_\alpha(F)$, $\varphi_\alpha(G)$ are linearly independent.
    \item If $F\in R$ is a square-free form, then 
    $\varphi_\alpha(F)$ does not have multiple factors other than $z^k$ for some $k\in \bN$.
\end{enumerate}
\end{proposition}

The next proposition corresponding to \cite[Proposition 6.9]{OS24}, tells us that if a finite set of forms resulting from a general quotient has small vector space dimension, then it must be the case that the original set of forms must have small vector space dimension.
We refer to this result as "lifting from general quotients," as we are lifting our upper bounds for the quotiented configurations to the original configurations.

\begin{proposition}[Lifting from general quotients]\label{proposition: lifting general quotient}
Let $d,e\in \bN$ such that $1\leq d\leq e$. Let $U\subset S_{\leq e}$ be a graded vector space generated by forms $H_1,\cdots,H_t$. Let $R=S/(U)$.
Let $V \subset R_{\leq e}$ be a $B$-lifted strong vector space with basis $F_1,\cdots,F_n\in R$.
Let $\varphi_\alpha :R[z]\rightarrow R[z]/I_\alpha$ be a graded quotient as defined in \cref{definition: quotients}.
Let $\cF\subset R_{\leq d}$ be a finite set of homogeneous elements.
Suppose that there exists $D\in \bN$ and a dense set $\cU \subset \bK^{n}$ such that $\dim \Kspan{\varphi_\alpha(\cF)}\leq D$ for every $\alpha \in \cU$. 
Then  
$$\dim\Kspan{\cF}\leq d^2(1+d)^{2n+2}D\cdot\Pi_{i=1}^t\deg(H_i)\cdot\Pi_{j=1}^n\deg(F_j).$$
\end{proposition}

While the statement of \cite[Proposition~6.9]{OS24} requires that the set of $\alpha$ for which the rank bound holds is an open set, the proof only requires the weaker condition that the set of such $\alpha$ is dense.
The fact that this weaker condition suffices will be crucial for us.

\subsection{Compatible graded quotients and twisted quotients}

\begin{definition}[Compatible graded quotients]\label{def: compatible quotients}
    Let $R=S/(U)$ and $V, Y\subseteq R$ be finite dimensional $\bK$-vector spaces such that $V\subseteq Y$. Let $\cF_V,\cF_Y$ be $\bK$-linear bases of $V,Y$ respectively. We say that $\cF_V,\cF_Y $ are compatible bases if $\cF_V \subseteq \cF_Y$. 

   Consider the compatible bases given by $\cF_V=\{F_1,\cdots,F_m\}$ and $\cF_Y=\{F_1,\cdots,F_m,G_1,\cdots,G_n\}$. Let $\alpha\in \bK^m$,$\beta\in \bK^n$ and let $y_1,y_2$ be new variables. Let $\varphi_{V,\alpha}:R[y_1]\rightarrow R[y_1]/I_{V,\alpha}$ be the graded quotient given by $F_i\mapsto \alpha_iy_1^{\deg(F_i)}$. Let $\varphi_{Y,(\alpha,\beta)}:R[y
   _1]\rightarrow R[y_1]/I_{Y,(\alpha,\beta)}$ be the graded quotient that sends $F_i\mapsto \alpha_iy_1^{\deg(F_i)}$ and $G_j\mapsto \beta_jy_1^{\deg(G_j)}$. Then we will say that $\varphi_{V,\alpha}$ and $\varphi_{Y,(\alpha,\beta)}$ are compatible graded quotients.
   
   \end{definition}

Given two compatible quotients $\varphi_{V,\alpha}$ and $\varphi_{Y,(\alpha,\beta)}$, we would like to express $\varphi_{Y,(\alpha,\beta)}$ as a composition of   $\varphi_{V,\alpha}$  and another graded quotient of the form $\varphi_{W,\beta}$, where $W\subseteq R[y_1]/I_{V,\alpha}$ is some vector space. However, in order to make such a composition work, we need to include the variable $y_1$ in the vector space $W$, and we will need to adjoin a new variable $y_2$ for the second quotient map. This leads to an incompatibility where it is not not enough to use a quotient of the form $\varphi_{W,\beta}$, since we also need to send $y_1\mapsto \gamma y_2$, for some scalar $\gamma\in \bK$, under the second quotient. This naturally leads to the notion of a twisted quotient defined below. Moreover, in \cref{proposition: general twists} we note that we can indeed decompose  $\varphi_{Y,(\alpha,\beta)}$ into a composition of $\varphi_{V,\alpha}$ and a twisted quotient.
\begin{definition}[Twists] Let $\gamma\in \bK^*$ and $d_1,\cdots,d_n\in\bN$. For $\beta\in\bK^n$, we define the twist of $\beta$ by $\gamma$ as the natural action of the torus given by $\gamma \cdot \beta=(\gamma^{d_1}\beta_1,\cdots,\gamma^{d_n}\beta_n)$. 

\end{definition}

\begin{definition}[Twisted quotient]
Let $V\subseteq R$ with a basis $\{F_1,\cdots,F_n\}$ and $\beta\in \bK^n$. For $\gamma\in\bK^*$, we will say that the graded quotient $\varphi_{V,\gamma\cdot \beta}$ is the $\gamma$-twist of the graded quotient $\varphi_{V,\beta}$.
    
\end{definition}
    
We will often consider a tuples $(\alpha,\beta,\gamma)$ where $\beta$ varies in an open subset depending on $\alpha$. Furthermore, $\gamma$ varies in an open subset depending on $\beta$. We note the following elementary result that will help us formalize this notion of relative generality.
\begin{proposition}\label{proposition: general twists}
    Fix integers $d_1,\cdots,d_n\geq 1$. Let $\cV\subseteq \bK^{n+1}$ be a non-empty open subset. Let $f:\bK\times \bK^n \rightarrow \bK^{n+1}$ be the map $f(\gamma,\beta)=(\gamma, \gamma^{d_1}\beta_1,\cdots,\gamma^{d_n}\beta_n)$. The we have the following.
    
    \begin{enumerate}
        \item The inverse image $f^{-1}(\cV)$ is an open subset of $\bK^{n+1}$.
        \item The image $p_2(f^{-1}(\cV))\subseteq \bK^n$ is a non-empty open subset, where $p_2$ is the projection map $\bK\times \bK^n \rightarrow \bK^n$.
        \item For each $\beta\in \cU_2:=p_2(f^{-1}(\cV))$, there exists a non-empty open subset $\cT_\beta\subseteq \bK$, such that  $f(\gamma,\beta)\in \cV$ for all $\gamma\in \cT_\beta$.
    \end{enumerate}  
\end{proposition}

\begin{proof}
    We know that $f$ is a continuous map in Zariski topology. Therefore $f^{-1}(\cV)$ is an open subset. We have the second statement since projection maps are open maps with respect to Zariski topology. Let $\cU_2:=p_2(f^{-1}(\cV))$. For each $\beta\in \cU_2$, we define $\cT_\beta:=f^{-1}(\cV)\cap p_2^{-1}(\beta)$. Note that the fiber $p_2^{-1}(\beta)$ is isomorphic to the affine space $\bK$. Since $\beta\in p_2(f^{-1}(\cV))$, we have that $\cT_\beta$ is non-empty. Furthermore, $\cT_\beta$ is open by the openness of $f^{-1}(\cV)$.
\end{proof}

\begin{proposition}\label{prop: composition of quotient}
    Let $R=S/(U)$ and $V, Y\subseteq R$ be finite dimensional $\bK$-vector spaces such that $V\subseteq Y$. Let $\alpha\in \bK^{\dim(V)}$ and $\beta\in \bK^{\dim(Y)-\dim(V)}$. Let $\varphi_{V,\alpha}$ and $\varphi_{Y,(\alpha,\beta)}$ be compatible graded quotients with compatible bases $\cF_V\subseteq \cF_Y$. Let $Y'=\Kspan{\cF_Y\setminus \cF_V}$. We denote $Y'_\alpha=\varphi_{V,\alpha}(Y')\subseteq R'_\alpha$, where  $R'_\alpha:=R[y]/I_{V,\alpha}$. Let $Y_\alpha=\Kspan{y_1,Y'_\alpha}$. Then the following holds.
    \begin{enumerate}
        \item Let $B:\bN^d\rightarrow \bN^d$ be such that $B_k(\delta+\mu)\geq B_k(\delta)+\norm{\mu}_1$ for all $\delta,\mu \in \bN^d$. If $Y$ is $B$-lifted strong, then $Y'_\alpha$ is also $B$-lifted strong.
        \item Let $\gamma\in \bK^*$. For a general $\alpha$, we have a commutative diagram of $\bK$-algebra homomorphisms.

        \[\begin{tikzcd}
 R[y_1] \arrow[r, "\varphi_{V,\alpha}"] \arrow[d, "id"] & R'_\alpha\arrow[rr,"\varphi_{Y_\alpha,\gamma\cdot(1,\beta)}"] & & R'_\alpha[y_2]/I_{Y_\alpha,\gamma\cdot(1,\beta)} \arrow[d,"\simeq"]\\
 R[y_1] \arrow[rrr, "\varphi_{Y,(\alpha,\beta)}"] & & &   R[y_1]/I_{Y,(\alpha,\beta)} \\
\end{tikzcd}\] 
where the map $\varphi_{Y_\alpha,\gamma\cdot(1,\beta)}$ sends $y_1\mapsto \gamma  y_2$ and on $Y'_\alpha$ it is same as the twisted quotient $\varphi_{Y'_\alpha,\gamma\cdot \beta}$. Moreover, the vertical isomorphism sends $y_2\mapsto \frac{1}{\gamma} y_1$.
    \end{enumerate}
\end{proposition}

\begin{proof}
    (1) If $Y=\oplus_{k=1}^d Y_k$ is $B$-lifted strong in $R$,  then we have $\lsmin{U,Y_k}\geq B_k(\delta_U+\delta_Y)$ for any graded piece $Y_k$.  Let $Q_1,\cdots,Q_r\in S$ lifts of a basis of $Y_k$. Now $\delta_Y=\delta_V+\delta_{Y'}$ and hence $\lsmin{U,Y_k}\geq B_k(\delta_U+\delta_{Y'})+\norm{\delta_V}_1$. Let $H_1,\cdots,H_t\in S$ be lifts of a basis of a graded piece $(Y'_\alpha)_k\subseteq R'_\alpha$. Then, for any $F\in U_k+\Kspan{H_1,\cdots,H_t}$, we may write $F=P+Q+T$, where $P\in U_k$, $Q\in \Kspan{Q_1,\cdots,Q_r}$ and $T\in I_{V,\alpha}$. Therefore, \[s(F)\geq s(P+Q)-s(T)\geq \lsmin{U,Y_k}-\norm{\delta_V}_1\geq B_k(\delta_U+\delta_{Y'})\geq B_k(\delta_U+\delta_{Y'_\alpha})\] 
    as $s(T)\leq\norm{\delta_V}_1 $ and  $\delta_Y\geq \delta_{Y'_k}$. Since $F\in U_k+\Kspan{H_1,\cdots,H_t}$ was arbitrary, we get that $\lsmin{U,(Y'_\alpha)_k}\geq B_k(\delta_U+\delta_{Y'_\alpha})$, and hence $Y'_\alpha$ is $B$-lifted strong.
    
    (2) Let $\cF_V=\{F_1,\cdots,F_m\}$ and $\cF_Y=\{F_1,\cdots,F_m,G_1,\cdots,G_n\}$. For a general $\alpha$, we have that the images of $G_1,\cdots,G_n$ in $R'_\alpha$ are linearly independent, and hence they are a basis of  $Y'_\alpha$. Then composition of  the first two horizontal maps sends $F_i\mapsto \alpha_i(\gamma y_2)^{\deg(F_i)}$ and $G_j\mapsto \beta_j(\gamma y_2)^{\deg(G_j)}$. Therefore, we have an isomorphism $R'_\alpha[y_2]/I_{Y_\alpha,\gamma\cdot(1,\beta)} \simeq R[y_1,y_2]/I$, where $I$ is generated by the forms $y_1-\gamma y_2, F_i-\alpha_iy_1^{\deg(F_i)},G_j- \beta_j(\gamma y_2)^{\deg(G_j)}$ with all $i\in [m], j\in[n]$.
    Now we have an isomprphism $R[y_1,y_2]/(y_1-\gamma y_2)\xrightarrow{\sim} R[y_1]$ given by $y_2\mapsto \frac{1}{\gamma}y_1$. This isomorphism descends to an isomorphism $R[y_1,y_2]/I\simeq R[y_1]/I_{Y,(\alpha,\beta)}$. By composing these two isomorphisms, we get the desired vertical isomorphism $R'_\alpha[y_2]/I_{Y_\alpha,\gamma\cdot(1,\beta)}\xrightarrow{\sim} R[y_1]/I_{Y,(\alpha,\beta)}$.
\end{proof}

\subsection{Graded Quotients and Edelstein-Kelly Configurations}\label{subsection: quotients and EK configurations}

We now define the image of a $(d, z, R)$-EK configuration under a graded quotient map.
We also show that the image defined this way is itself an EK configuration.

\begin{definition}\label{definition: quotient ek configuration}
    Let $U \subset S$ be a graded finitely generated vector space such that $R := S / \ideal{U}$ is a CM UFD.
    Let $\cA, \cB, \cC$ be a $(d, z, R)$-EK configuration.
    Let $V \subset R_{\leq d}$ be a $h_{B} \circ t_{2}$-lifted strong vector space such that $z \in V$, and $\varphi_{\alpha}: R\bs{y} \to R\bs{y} / I_{\alpha}$ be a graded quotient.

    We define a configuration $\varphi_{\alpha}(\cA), \varphi_{\alpha}(\cB), \varphi_{\alpha}(\cC)$ as
    $$ \varphi_{\alpha}(\cA) := \setbuild{A'}{A \in \cA, A' |  \varphi_{\alpha}(A), A' \neq 1, (A', y) = 1, A' \textrm{ irreducible}},$$    
    $$ \varphi_{\alpha}(\cB) := \setbuild{B'}{B \in \cB, B' |  \varphi_{\alpha}(B), B' \neq 1, (B', y) = 1, B' \textrm{ irreducible}},$$   
    $$ \varphi_{\alpha}(\cC) := \setbuild{C'}{C \in \cC, C' |  \varphi_{\alpha}(C), C' \neq 1, (C', y) = 1, B' \textrm{ irreducible}}.$$   
    
\end{definition}

\begin{proposition}
    \label{prop: graded quot of ek}
    Let $U \subset S$ be a graded finitely generated vector space such that $R := S / \ideal{U}$ is a CM UFD.
    Let $\cA, \cB, \cC$ be a $(d, z, R)$-EK configuration.
    Let $V \subset R_{\leq d}$ be a $h_{B} \circ t_{2}$-lifted strong vector space such that $z \in V$, and $\varphi_{\alpha}: R\bs{y} \to R\bs{y} / I_{\alpha}$ be a graded quotient.
    Let $\varphi_{\alpha}(\cA), \varphi_{\alpha}(\cB), \varphi_{\alpha}(\cC)$ be the image of the EK configuration under the graded quotient as defined in \cref{definition: quotient ek configuration}.
    There is an open subset $\cU \subset \bK^{\dim V}$ such that for every $\alpha \in \cU$, the sets $\varphi_{\alpha}(\cA), \varphi_{\alpha}(\cB), \varphi_{\alpha}(\cC)$ are an $(d, y, R\bs{y}/I_{\alpha})$-EK configuration.
\end{proposition}

\begin{proof}
    Let $\cF := \cA \cup \cB \cup \cC$ and $\varphi_{\alpha}(\cF) := \varphi_{\alpha}(\cA) \cup \varphi_{\alpha}(\cB) \cup \varphi_{\alpha}(\cC)$.
    We start by verifying the first condition for the configuration to be an EK configuration, which is equivalent to $\varphi_{\alpha}(\cF) \cup \bc{y}$ consisting of pairwise non associate irreducible forms.

    By construction, the forms in $\varphi_{\alpha}(\cF) \cup \bc{y}$ are irreducible.
    By item 3 of \cref{proposition: gcd after projection}, for any $F$ and $F'$ such that $\varphi_{\alpha}(F) = y^{r} F'$, we have that $F'$ is squarefree.
    Suppose $F', G' \in \varphi_{\alpha}(\cF) \setminus \bc{y}$ are irreducible factors of $\varphi_{\alpha}(F), \varphi_{\alpha}(G)$ respectively, for $F, G \in \cF$.
    By item 2 of \cref{proposition: gcd after projection} we have $(\varphi_{\alpha}(F), \varphi_{\alpha}(G)) = y^{t}$ for some $t$.
    Since $(F', G') | (\varphi_{\alpha}(F), \varphi_{\alpha}(G))$, and since $(y, F') = (y, G') = 1$, we deduce that $(F', G') = 1$.
    Therefore the forms in $\varphi_{\alpha}(\cF) \cup \bc{y}$ are pairwise non associate.

    Suppose $A' \in \varphi_{\alpha}(\cA), B' \in \varphi_{\alpha}(\cB)$ are irreducible factors of $\varphi_{\alpha}(A), \varphi_{\alpha}(B)$ respectively, for $A \in \cA, B \in \cB$.
    We have
    $$z \cdot \prod_{C \in \cC} C \in \radideal{A, B}.$$
    Applying the map $\varphi_{\alpha}$, we obtain
    $$y \cdot \prod_{C\in \cC} y^{t_{C}} \varphi_{\alpha}(C) \in \radideal{\varphi_{\alpha}(A), \varphi_{\alpha}(B)} \subset \radideal{A', B'}.$$
    For each such $C$, we either have $\varphi_{\alpha}(C) = y^{t_{C}}$ for some $t_{C}$, or $C' | \varphi_{\alpha}(C)$ for some $C' \in \varphi_{\alpha}(\cC)$.
    Therefore we have
    $$y \cdot \prod_{C' \in \varphi_{\alpha}(\cC)} C' \in \radideal{A', B'}.$$
    This shows the second condition for EK configurations, and a symmetric argument shows the third and fourth conditions.
\end{proof}

\subsection{Potential function}\label{subsection: potential function}

In this section, we define a potential function that will be useful in controlling EK configurations.
Informally, the potential function measures how close a form is to a given vector space.
It will be convenient to define the potential so that the closer the form is to the space, the higher the potential is.
Our main arguments will therefore involve steps that increase the potential by a sufficient amount.

Let $U \subseteq S$ be a graded finitely generated vector space such that $R := S / \ideal{U}$ is a CM UFD.
Let $V \subseteq R$ be a $h_{2B} \circ t_{1}$-lifted strong vector space and $\alpha \in \bK^{\dim(V)}$.
The quotient map $\varphi_{V,\alpha}: R[y] \rightarrow R[y] / I_{V,\alpha}$ is defined in \cref{sec:generalquot}. 

\emph{Individual potential.} For a form $F \in R$, let $\varphi_{V,\alpha}(F)=y^{r} \prod_{i=1}^{s}F_{i}^{r_i}$ be an irreducible factorization.
We define the potential of $F$ with respect to $V$ and $\alpha$ as 
\[ \psi_{V,\alpha}(F)=2r+s. \]

 \begin{proposition}
[Potential of product] \label{prop: valuative property}For any two forms $F,G\in R$. Suppose that $\varphi_{V,\alpha}(F),\varphi_{V,\alpha}(G)$ have $c$ number of common irreducible factors other than $y$. Then
\[\Psi_{V,\alpha}(FG)=\Psi_{V,\alpha}(F)+\Psi_{V,\alpha}(G)-c.\]
In particular, 
\[\Psi_{V,\alpha}(FG)\leq \Psi_{V,\alpha}(F)+\Psi_{V,\alpha}(G).\]
Moreover, equality occurs iff $\gcd(F,G)=y^t$ for some $t\geq 0$.
 \end{proposition}

 \begin{proof}
     Suppose we have the irreducible factorizations $\varphi_{V,\alpha}(F)=y^{a} \prod_{i=1}^{s}F_{i}^{d_i}$ and $\varphi_{V,\alpha}(G)=y^{b} \prod_{i=1}^{t}G_{i}^{e_i}$.  Since $\varphi_{V,\alpha}(FG)=\varphi_{V,\alpha}(F)\varphi_{V,\alpha}(G)$, we have
     \[\Psi_{V,\alpha}(FG)=2(a+b)+s+t-c\leq 2a+s+2b+t=\Psi_{V,\alpha}(F)+\Psi_{V,\alpha}(G).\]
     Moreover, we have equality iff $c=0$.
 \end{proof}
 
\emph{Group potential.} For a finite set of forms $\cG \subseteq R$, we define the potential of $\cG$ with respect to the pair $(V,\alpha)$ as 
\[\Phi_{V,\alpha}(\cG)=\sum_{G\in\cG}\psi_{V,\alpha}(G).\]

\begin{remark}\label{remark: group potential bound}
    Note that if $\cG\subseteq R_{\leq d}$, then $\Phi_{V,\alpha}(\cB)\leq 2d|\cB|$ for any $(V,\alpha)$.
\end{remark}

\emph{Potential increasing forms.} Let $V\subseteq Y$ be both $h_{2B}\circ t_1$-lifted strong vector spaces in $R$. Let $\alpha\in \bK^{\dim(V)}$ and $\beta\in \bK^{\dim(Y)-\dim(V)}$. We say that $F\in \cG$ is a \textit{potential increasing} form with respect to the pairs $(V,\alpha)$ and $(Y,(\alpha,\beta))$ if 
\[\Psi_{Y,(\alpha,\beta)}(F)> \Psi_{V,\alpha}(F).\]

The notion of potential increasing captures the property that a form factors more when quotiented with respect to $Y$ than $V$. First we note that for general quotients potential does not decrease.

\begin{proposition}[Non-decreasing property]\label{prop: potential non-decreasing}
    Let $V\subseteq Y$ be both $h_{2B}\circ t_1$-lifted strong vector spaces in $R$ of dimensions $m$ and $m+n$ respectively. Let $\varphi_{V,\alpha}$ and $\varphi_{Y,(\alpha,\beta)}$ be compatible graded quotients with respect to bases $\cF_V\subseteq \cF_Y$. Let $Y'=\Kspan{\cF_Y\setminus \cF_V}$, $R'_{\alpha}=R[y_1]/I_{V,\alpha}$, $Y'_\alpha=\varphi_{V,\alpha}(Y')$ and $Y_\alpha=\Kspan{y_1,Y'_\alpha}$.
    
    Fix a square-free form $F\in R$. For a general $\alpha\in \bK^m$, there exists a non-empty open subset $\cU\subseteq \bK^n$ (depending on $\alpha$) such that for all $\beta\in \cU$ the following holds.    \begin{enumerate}
        \item We have $\Psi_{Y,(\alpha,\beta)}(F)\geq \Psi_{V,\alpha}(F).$
        \item Moreover, $\Psi_{Y,(\alpha,\beta)}(F)= \Psi_{V,\alpha}(F)$ iff one of the following holds
        \begin{enumerate}
            \item $F\in \bK[V]$.
            \item Let $F_i$ be an irreducible factor of  $\varphi_{V,\alpha}(F)$ such that $F_i\not\in (y_1)$. Then its image $\varphi_{Y_\alpha,\gamma\cdot(1,\beta)}(F_i)$ is irreducible and  $\varphi_{Y_\alpha,\gamma\cdot(1,\beta)}(F_i)\not \in (y_2)$ for general $\gamma \in \bK$. 
        \end{enumerate}
    \end{enumerate}
    
\end{proposition}

\begin{proof}
     Note that the quotient map $\varphi_{Y,(\alpha,\beta)}:R\rightarrow R[y_2]/I_{Y,(\alpha,\beta)}$ factors, up to isomorphism, as $\varphi_{Y,(\alpha,\beta)}=\varphi_{Y_\alpha,\gamma\cdot(1,\beta)}\circ\varphi_{V,\alpha}$ by \cref{prop: composition of quotient}. Suppose we have the irreducible factorization  \[\varphi_{V,\alpha}(F)=y_1^{r} \prod_{i=1}^{s}F_{i}^{r_i}.\] Then we have \[\varphi_{Y,(\alpha,\beta)}(F)=y_1^{r} \prod_{i=1}^{s}\varphi_{Y_\alpha,\gamma\cdot(1,\beta)}(F_{i})^{r_i}.\]

    If $s=0$, then $\varphi_{V,\alpha}(F)$ is a power of $y_1$ up to scalar multiples. By \cref{proposition: general quotient general}, we have $V\in \bK[V]$. In particular, $\Psi_{Y,(\alpha,\beta)}(F)= \Psi_{V,\alpha}(F)=2\deg(F)$.

    Suppose that $s>0$. By \cref{proposition: gcd after projection}, there exists a non-empty open subset $\cV\subseteq \bK^{n+1}$ such that for all $\gamma\cdot (1,\beta)\in \cV$, the forms $\varphi_{Y_\alpha,\gamma\cdot(1,\beta)}(F_i)$ and $\varphi_{Y_\alpha,\gamma\cdot(1,\beta)}(F_j)$ do not have any common factors other than $y_2$ for $i\neq j$. We let $\cU$ be the image of $f^{-1}(\cV)$ under the projection  $\bK\times \bK^n\rightarrow \bK^n$, where $f(\gamma,\beta)=\gamma\cdot(1,\beta)$ as in \cref{proposition: general twists}. Let $G_i=\varphi_{Y_\alpha,\gamma\cdot(1,\beta)}(F_i)$. We have an irreducible factorization $G_i=y_2^{a_i} \prod_{i=1}^{s_i}H_{i}$. Note that since $G_i$ is irreducible, the forms $H_i$ appear with multiplicity $1$ by \cref{proposition: gcd after projection}. Now, using the isomorphism in \cref{proposition: general twists}, we have
    \[\Psi_{Y,(\alpha,\beta)}(F)=2r+\sum_{i=1}^s 2r_ia_i+\sum_{i=1}^s s_i\geq 2r+s=\Psi_{V,\alpha}(F).\]
    Moreover, equality occurs iff $a_i=0$ and $s_i=1$ for all $i\in [s]$. In other words $G_i$ is irreducible and $G_i$ not a scalar multiple of $y_2$.
\end{proof}

We note the following criterion for when the potential strictly increases. 

\begin{lemma}[Potential increasing criterion]\label{lemma: abs red and potential increase}
    Let $V\subseteq Y$ be both $h_{2B}\circ t_1$-lifted strong vector spaces in $R$ of dimensions $m$ and $m+n$ respectively. Let $\varphi_{V,\alpha}$ and $\varphi_{Y,(\alpha,\beta)}$ be compatible graded quotients with respect to bases $\cF_V\subseteq \cF_Y$. Let $Y'=\Kspan{\cF_Y\setminus \cF_V}$, $R'_{\alpha}=R[y_1]/I_{V,\alpha}$, $Y'_\alpha=\varphi_{V,\alpha}(Y')$ and $Y_\alpha=\Kspan{y_1,Y'_\alpha}$.
    
    Let $F\in R$ be a square-free form. Let $\alpha\in \bK^m$ be a general element. Suppose that there exists a factor $F_k$ of $\varphi_{V,\alpha}(F)$ in $R'_\alpha$, such that
    \begin{enumerate}
        \item $F_k$ is not a scalar multiple of $y_1$, and
        \item $F_k$ is absolutely reducible with respect to $Y_\alpha$ or $F_k$ is in the ideal $(Y_\alpha)$ in $R'_\alpha$.
    \end{enumerate}

    Then there exists a non-empty open subset $\cU\subseteq \bK^n$ (depending on $\alpha$) such that for all $\beta\in \cU$, the form $F$ is potential increasing with respect to the pairs $(V,\alpha)$ and $(Y,(\alpha,\beta))$.
\end{lemma}

\begin{proof}
    Suppose that we have factorization $\varphi_{V,\alpha}(F)=y_1^{r} \prod_{i=1}^{s}F_{i}^{r_i}$ in $R'_\alpha$.  Note that the quotient map $\varphi_{Y,(\alpha,\beta)}:R\rightarrow R[y_2]/I_{Y,(\alpha,\beta)}$ factors as $\varphi_{Y,(\alpha,\beta)}=\varphi_{Y_\alpha,\gamma\cdot(1,\beta)}\circ\varphi_{V,\alpha}$ by \cref{prop: composition of quotient}. Let $\cU\subseteq \bK^n$ be the non-empty open subset provided by $\cref{prop: potential non-decreasing}$. 
    
    Let $F_k$ be a factor of $\varphi_{V,\alpha}(F)$ with properties $(1)$ and $(2)$ in the hypothesis. In particular, $F_k$ is not a scalar multiple of $y_1$. Therefore $F\not \in \bK[V]$ by \cref{proposition: general quotient general}.
    
    \emph{Case 1.} Suppose that $F_k\in (Y_\alpha)$. Then for any $\gamma,\beta$, the form $\varphi_{Y_\alpha,\gamma\cdot(1,\beta)}(F_k)$ is divisible by $y_2$. Therefore $\varphi_{Y_\alpha,\gamma\cdot(1,\beta)}(F_k)$ is reducible if  $\deg(\varphi_{Y_\alpha,\gamma\cdot(1,\beta)}(F_k))>1$. If $\deg(\varphi_{Y_\alpha,\gamma\cdot(1,\beta)}(F_k))=1$, then it is a scalar multiple of $y_2$.

    \emph{Case 2.} Now suppose that $F_k$ is absolutely reducible with respect to $Y_\alpha$. Then by \cref{proposition: general quotient general}, there exists  a non-empty open subset $\cV\subseteq \bK^{n+1}$ such that $\varphi_{Y_\alpha,\gamma\cdot(1,\beta)}(F_k)$ is a reducible form for $\gamma\cdot(1,\beta)\in \bK^{n+1}$. We let $\cU_2$ to be the image of $f^{-1}(\cV)$ under the projection  $\bK\times \bK^n\rightarrow \bK^n$, where $f(\gamma,\beta)=\gamma\cdot(1,\beta)$ as in \cref{proposition: general twists}.
    
    Therefore, by combining both cases above,  there exists a non-empty open subset $\cU':=\cU_2\cap \cU$ such that for all $\beta\in \cU'$, the form $\varphi_{Y_\alpha,\gamma\cdot(1,\beta)}(F_k)$ is either reducible or it is a scalar multiple of $y_2$, for a general $\gamma$ depending on $\beta$. Then \cref{prop: potential non-decreasing} implies that we must have a strict inequality $\Psi_{Y,(\alpha,\beta)}(F)> \Psi_{V,\alpha}(F)$ for all $\beta\in\cU'$.
\end{proof}

\section{Uniform bounds on Edelstein-Kelly type configurations}\label{section: EK theorem}

In this section, we show that higher degree EK configurations have bounded rank.
As discussed earlier, our approach is to induct on the degree of the forms in the configuration.
The base case is a configuration where all the forms are linear.
We show that such configurations are essentially $(1, 1)$-linear-EK configurations, therefore they have bounded rank by \cref{prop:linearek}.
In the inductive step, we control all the forms of highest degree in the configuration, and then apply a graded quotient to reduce to configurations where the highest degree forms have degree one lower than what we started with.

In the first subsection, we focus on controlling the highest degree forms in EK-configurations.
The second subsection will then formalize the above inductive argument.
For the rest of this section, fix positive integer $\eta \geq 3$ and set $\varepsilon := 1/16$.
We will use the auxiliary functions $H, D$ defined in \cref{sec:strongalg}.

\subsection{Controlling highest degree forms}

We begin by establishing some notation specific to this subsection.

\begin{notation}
    Fix an integer $e$.
Let $\strong: \bN^{e} \to \bN^{e}$ be an ascending function such that $\strong_{i}(\delta) \geq A(\eta, i) + 3 \norm{\delta}_{1}$, where $A$ is the function defined in \cite[Theorem~A]{AH20}.
Let $U \subseteq   S_{\leq e}$ be a graded finitely generated vector space such that $R := S / \ideal{U}$ is a UFD, and $z \in R_{1}$.
Let $\cA, \cB, \cC \subset R_{\leq d}$ be sets that form a $(d, z, R)$-EK configuration (\cref{def:ek}) for $2 \leq d \leq e$.
Let $\cF := \cA \cup \cB \cup \cC$.

Suppose $V \subseteq R_{\leq d}$ is a $h_{2 \strong} \circ t_{1}$-lifted strong vector space.
Let $\cA_{d}^{V}$ denote the set of forms in $\cA_{d}$ which are not in $(V)$ and are absolutely irreducible with respect to $V$.

\[\cA_d^V:=\{F\in \cA_d\mid F\not\in (V) \text{ and } F \text{ absolutely irreducible with respect to } V\}.\]
We denote $m^{V}_{a}:=\abs{\cA_{d}^{V}}$.
We use $M^{V}_{a}$ to denote $\abs{\cA} - m^{V}_{a}$, this is the number of forms in $\cA$ that either have degree less than $d$, or are absolutely reducible over $V$, or are in the ideal generated by $V$.

In particular, if we take $V=(0)$, then $m^{\br{0}}_{a}, M^{\br{0}}_{a}$ are the sizes of the degree $d$ and degree less than $d$ pieces of $\cA$ respectively.
We similarly define $m_{b}, M_{b}, m_{c}, M_{c}$.
We will assume that $\abs{\cA} \geq \abs{\cB} \geq \abs{\cC}$.

\end{notation}

Our goal in this section is to construct a vector space $W$ such that all the degree $d$ forms in $\cA,\cB,\cC$ are absolutely reducible over $W$ or in the ideal $(W)$, i.e. $m_a^W=m_b^W=m_c^W=0$. First, we show that if the algebra $\bK[V]$ contains a sufficient fraction of forms, then we may find $W$ such that at least one of $m_a^W,m_b^W,m_c^W$ is $0$.

\begin{lemma} \label{lem: ek algebra control}
    Suppose $\abs{\cC} \geq 2 \irredprimebound{d} \varepsilon^{-1}$.
    Further suppose $V$ is a $H(\strong, 1, 2 \cdot d \cdot \varepsilon^{-1})$-lifted strong vector space such that one of the following holds.
    \begin{itemize}
        \item $\abs{\bK\bs{V} \cap \cA} \geq \varepsilon \abs{\cA}$ or 
        \item $\abs{\bK\bs{V} \cap \cB} \geq \varepsilon \abs{\cB}$ or
        \item $\abs{\cC} \geq \abs{\cB} / 2$ and $\abs{\bK\bs{V} \cap \cC} \geq \varepsilon \abs{\cC}$.
    \end{itemize}
    Then there is a $\strong$-lifted strong vector space $W$ obtained by applying a $(1, 4 \cdot d \cdot \varepsilon^{-1})$ process to $V$ such that $m_{b}^{W}\cdot m_{b}^{W}\cdot m_{c}^{W}=0$, i.e. at least one of the sets $\cA,\cB,\cC$ have no degree $d$ forms absolutely irreducible over $W$.
\end{lemma}

\begin{proof}
    We first assume $\abs{\bK\bs{V} \cap \cA} \geq \varepsilon \abs{\cA}$.
    Let $A_{1}, \dots, A_{r}$ be the elements in $\abs{\bK\bs{V} \cap \cA}$.
    In this case a $(1, 2 \cdot d \cdot \varepsilon^{-1})$ process will suffice.

    We now describe one iteration of our $(1, 2 \cdot d \cdot \varepsilon^{-1})$ process.
    Let $B \in \cB_{d}$ be a form that is absolutely irreducible over $V$.
    If no such $B$ exists, then every form in $\cB_{d}$ is absolutely reducible over $V$, which implies $m_{b}^{V} = 0$, and we terminate the process.
    Since $B$ is absolutely irreducible over $V$, there are at most $\primebound{d}$ elements among $A_{1}, \dots, A_{r}$ such that $\ideal{B, A_{i}}$ is not prime.
    Reordering, we can assume that $\ideal{A_{i}, B}$ is prime for $1 \leq i \leq s$ with $s = r - \primebound{d}$.
    Since $r \geq \varepsilon \abs{\cA}$ and $\abs{\cA} \geq \abs{\cC} \geq 2 \primebound{d} \varepsilon^{-1}$, we have $s \geq \varepsilon \abs{\cA} / 2$.

   Note that $z\not \in (A_i,B)$ as $\deg(B)=d>1$ and $z,A_i$ are non-associate forms. Hence, by the EK condition, there exists $C_{i} \in \ideal{A_{i}, B} \cap \cC$ for $1 \leq i \leq s$, as $(A_i,B)$ is prime.
    Since $\deg{B} = d$, we have $\deg{C_{i}} = d$, and we can write $C_{i} = \alpha_{i} B + H_{i} A_{i}$.
    If $C_{i} \in \ideal{V}$ then $B \in \ideal{V}$, contradicting assumption.
    If $C_{i} = C_{j}$, then $\alpha_{j} B + H_{j} A_{j} = \alpha_{i} B + H_{i} A_{i}$.
    If we also have $\alpha_{i} \neq \alpha_{j}$ then again $B \in \ideal{V}$, a contradiction.
    Therefore, if $C_{i} = C_{j}$ then $\alpha_{i} = \alpha_{j}$, and further $A_{i}, A_{j}$ are distinct factors of $C_{i} - \alpha_{i} B = C_{j} - \alpha_{j} B$.
    Since $\deg  C_{j} - \alpha_{j} B = d$, and since these factors are distinct, there are at most $d$ indices such that $C_{i_{1}} = \cdots = C_{i_{a}}$.
    In conclusion, the above paragraph shows that the set of forms $C_{1}, \dots, C_{s}$ consists of at least $s / d$ distinct forms, none of which are in the ideal $\ideal{V}$.

    We now update $V$ to $AH_R(V, B)$.
    After this update, $m_{c}^{V}$ reduces by at least $s / d$, which is at least $\varepsilon \abs{\cA} / 2d$.
    Note that we still have $|\bK[V]\cap A|\geq \varepsilon |A|$. Therefore if $m_{b}^{V}\neq 0$, we may apply the argument above again and reduce $m_{c}^{V}$ by $s/d$.
    Since $m_{c}^{V}$ is at most $\abs{\cA}$ to start with, after $2 \cdot d \cdot \varepsilon^{-1}$ complete iterations of the above, we have a vector space $W$ such that $m_{b}^{W}=0$ or $m_{c}^{W} = 0$.
    The claimed bound on the lifted strength of $W$ follows from \cref{lem: iterated lifted strength}.
    
    In the case when $\abs{\bK\bs{V} \cap \cB} \geq \varepsilon \abs{\cB}$ the exact same argument applies with every occurrence of $\cA$ and $\cB$ swapped.
    In the last case, we have $\abs{\bK\bs{V} \cap \cC} \geq \varepsilon \abs{\cC}$ and $\abs{\cC} \geq \abs{\cB} / 2$.
    If we just swap the roles of $\cA$ and $\cC$ in the above argument, then essentially in each iteration we are reducing $m_{a}^{V}$ by $\varepsilon \abs{\cC} / 2$.
    However, this does not suffice, because $m_{a}^{V}$ might be significantly bigger than $\cC$, and therefore we cannot guarantee that the iteration terminates after constantly many steps.
    Therefore, instead of just swapping the roles of $\cA$ and $\cC$, we replace the role of $\cA$ with $\cC$ and the role of $\cB$ with $\cA$.
    Now in each iteration we are reducing $m_{b}^{V}$ by at least $\varepsilon \abs{\cC} / 2$.
    The further assumption that $\abs{\cC} \geq \abs{\cB} / 2$ implies that we are reducing $m_{b}^{V}$ by at least $\varepsilon \abs{\cB} / 4$.
    Therefore the above iteration terminates after $4 \cdot d \cdot \varepsilon^{-1}$ steps.    
\end{proof}

\begin{definition}
    \label{def: unbreakable pencil}
    Let $F_1,\cdots,F_a\in R$ be homogeneous of degree $d$. We will say that $F_1,\cdots,F_a$   \textit{generate unbreakable pencils} if they are linearly independent and $\smin{F_i,F_j}>0$ for all $i,j\in [a]$. Alternatively, given a set of degree $d$ forms $\cS\subseteq R_d$, we will say that the set $\cS$ is \textit{unbreakable pencil generating} if $\cS$ is a linearly independent set and $\smin{F,G}>0$ for any $F,G\in \cS$.
\end{definition}

Note that if $F_1,\cdots,F_a$ generate unbreakable pencils, then $F_i$ is irreducible for all $i\in [a]$.

\begin{lemma}
\label{lem:pairwise strong sequence control}
Let $\cG$ be a finite set non-associate irreducible homogeneous forms in $R_{d}$. Suppose that for any unbreakable pencil generating subset $\cS \subseteq \cG$, we have $|\cS|\leq b$. 
Then there exists a subset $T\subseteq \cG$ such that the following holds.
\begin{enumerate}
    \item $\abs{T} \leq \br{\irredprimebound{d} + 1} \cdot b$. 
    \item Let $Y\subseteq R$ be any $h_{2\strong}\circ t_1$-lifted strong vector space with $T\subseteq \bK[Y]$. Then every form in $\cG$ is absolutely reducible over $Y$ or in the ideal $(Y)$. 
\end{enumerate}
\end{lemma}

\begin{proof}
    Let us consider the unbreakable pencil generating subsets of $\cG$ with respect to the partial order given by inclusion. Let $\cS_1\subseteq \cG$ be a maximal unbreakable pencil generating set. We will define subsets $\cH_i\subseteq \cG$ iteratively. We set $\cH_1:=\cS_1$. Note that for every form $F \in \cG \setminus \cH_{1}$ we have either $\smin{F, H_1} = 0$ for some $H_1\in \cH_1$ or that $F \in \ideal{\cH_1}$ the ideal generated by $\cH_1$. Now let $\cS_2 \subseteq \cG\setminus \cH_1$ be a maximal unbreakable pencil generating set. Note that if $\{F\}\cup \cS_2$ is not an unbreakable pencil generating set for some $F\in \cG$, then $F\in \cH_1$. Let $\cH_2:=\cH_1\cup \cS_2$. Then, for every form $F \not \in \cH_{2}$ there is a $H_2\in \cS_2$ such that $\smin{F, H_2} = 0$, or $F \in \ideal{\cS_2}$.
    
    We repeat this process iteratively and set $\cH_{i+1}:=\cH_i\cup \cS_{i+1}$ where $\cS_{i+1}$ is a maximal unbreakable pencil generating subset of $\cG\setminus \cH_i$. Note that $|\cS_{i}|\leq b$ for all $i$ by assumption.
    Therefore, in each iteration, we add at most $b$ forms.
    We do this $\irredprimebound{d} + 1$ times. Then we have $\cH_{\cC(d) + 1}\leq (\irredprimebound{d} + 1)b$. Let $T:=\cH_{\irredprimebound{d} + 1}$.
    We will show that this set $T$ satisfies the desired property.
    Let $Y \subseteq R$ be a $h_{2B}\circ t_1$-lifted strong vector space such that $T\subseteq \bK[Y]$.

    Let $F \in \cG$ be a form that is not in the ideal $\ideal{Y}$. Hence $F\not \in (\cS_i)$ for all $i$. Recall that $\cS_i$ is a maximal unbreakable pencil generating set in $\cG\setminus \cS_{i-1}$. 
    Therefore, for every $1 \leq i \leq \irredprimebound{d} + 1$, there is a $H_i\in \cS_{i}$ such that $\smin{F, H_i} = 0$.
    Let $P_{i} = F + \beta_{i} H_i$ with $s(P_{i}) = 0$.
    In particular, $P_i$ is reducible in $R$. Note that $P_i\not \in (Y)$, as $F\not \in (Y)$. Hence no irreducible factor of $P_i$ can lie in $\bK[Y]$. Hence, the forms $P_{i}$ are absolutely reducible over $Y$.
    If $P_{i} = \gamma P_{j}$, then $F \in \ideal{H_i, H_j}$ by the assumption that the forms in $\cG$ are non-associate, therefore $F \in \bK\bs{Y}$.
    If the forms $P_{i}, P_{j}$ are all non-associate and not in $(Y)$, then by \cref{cor:strongprimebound} the form $F$ must be absolutely reducible with respect to $Y$.
    This completes the proof.
\end{proof}

Recall that for $\alpha\in \bK^{\dim(V)}$, we denote by $\varphi_{V,\alpha}$ the graded quotient homomorphism $R[y]\rightarrow R[y]/I_\alpha$ as defined in \cref{sec:generalquot}. We will look for subsets $\cS\subseteq\cA_{d}^V$ such that $\varphi_{V,\alpha}(\cS)$ is an unbreakable pencil generating set in $R[y]/I_\alpha$. The following lemma shows that if all such subsets have their size upper bounded by a fixed constant, then we can increase the vector space $V$ to $W$ in a controlled way and ensure that $m_a^W=0$.

\begin{lemma}\label{lemma: breaking condition 2}
Fix an integer $r>0$.
Suppose $V$ is $H(\strong, 1, (\irredprimebound{d}+1)\cdot r)$-lifted strong.
Suppose that there exists a dense subset $\cZ\subseteq \bK^{\dim(V)}$ such that for all $\alpha\in \cZ$ we have

    \[r> \max\{|\cS|\mid \cS\subseteq \cA_d^V \text{ and } \varphi_{V,\alpha}(\cS)\text{ is an unbreakable pencil generating set}\}.\]
    Then there exists a $\strong$-lifted strong vector space $W$ obtained by applying a $(1, (\irredprimebound{d}+1)\cdot r)$-process to $V$ such that every form in $\cA_d^V$ is absolutely reducible over $W$ or in $(W)$, i.e. we have $m_a^W=0$.
\end{lemma}

\begin{proof}
    By \cref{proposition: gcd after projection}, we can further replace $\cZ$ by a possibly a possible smaller dense subset, and assume that $\varphi_{V,\alpha}(\cA^V_d)$ consists of pairwise non-associate forms. Then for each $\alpha\in \cZ$, we can apply \cref{lem:pairwise strong sequence control} to the set $\varphi_{V,\alpha}(\cA^V_d) \subseteq R[y_1]/I_{V,\alpha}$ and obtain a subset $T_\alpha\subseteq \varphi_{V,\alpha}(\cA^V_d)$ with properties $(1)$ and $(2)$ from \cref{lem:pairwise strong sequence control}. For each $\alpha\in \cZ$ let us choose $\widetilde{T}_\alpha\subseteq \cA_d^V$ which a lift of $T_\alpha$. Note that there are finitely many choices for $\widetilde{T}_\alpha$. Therefore we may replace $\cZ$ by a possibly smaller dense subset and assume that there exists a fixed set $\widetilde{T}\subseteq \cA_d^V$ such that $\widetilde{T}=\widetilde{T}_\alpha$ for all $\alpha\in \cZ$. In particular, the following holds:
    \begin{enumerate}
        \item $|\widetilde{T}|\leq(\irredprimebound{d}+1)\cdot r$, and 
        \item For all $\alpha\in \cZ$ and for any $Y_\alpha\subseteq R[y_1]/I_{V,\alpha}$ which is $h_{2\strong}\circ t_1$-lifted strong vector space with $\varphi_{V,\alpha}(\widetilde{T})\subseteq \bK[Y]$,  every form in $\varphi_{V,\alpha}(\cA_d^V)$ is absolutely reducible over $Y$ or in the ideal $(Y_\alpha)$ .
    \end{enumerate}  
    
    We let $W:=AH(V, \widetilde{T})$. Now we check that we will have $m_a^W=0$. Note that $V\subseteq W$ by \cref{corollary: practical robustness}.  Let $W=V+W'$, where $W'$ is obtained by extending a basis of $V$ to a basis of $W$. Then $\varphi_{V,\alpha}(W')$ is a $h_{2\strong} \circ t_2$-lifted strong vector space in $R[y]/I_\alpha$ by part (1) of \cref{prop: composition of quotient}. Let $W'_\alpha=\varphi_{V,\alpha}(W')$. We also have $\varphi_{V,\alpha}(\widetilde{T})\subseteq \bK[y_1,W'_\alpha]$ for all $\alpha\in \cZ$. In particular, by taking $W_\alpha:=\Kspan{y_1,W'_\alpha}$, we know that every form in $\varphi_{V,\alpha}(\cA_d^V)$ is absolutely reducible over $W_\alpha$ or in $(W_\alpha)$. 
    
    Let $F\in \cA_d^V$. Note that by replacing $\cZ$ with a possibly smaller dense subset, we may assume that $\dim(W')=\dim(W'_\alpha)$ for all $\alpha\in \cZ$ by \cref{proposition: gcd after projection}.  For $\beta\in \bK^{\dim(W')}$ and $\gamma\in \bK$, consider the map $\varphi_{Y_\alpha,\gamma\cdot(1,\beta)}$ as in \cref{prop: composition of quotient}. Let us denote $R'_\alpha=R[y]/I_{V,\alpha}$. Let $G:=\varphi_{V,\alpha}(F)$. We consider two situations below.

    \emph{Case 1.} Suppose that $G\in (W_\alpha)$ in $R[y_1]/I_{V,\alpha}$. Then we may compose with map $\varphi_{W_\alpha,\gamma\cdot (1,\beta)}$ as in \cref{proposition: general twists}. Under this composition $\varphi_{W_\alpha,\gamma\cdot (1,\beta)}\circ\varphi_{V,\alpha}$, we see that $F$ mapped to a form in $(y_2)$. By the isomorphism in \cref{prop: composition of quotient}, we conclude that $\varphi_{W,(\alpha,\beta)}$ maps $F$ to a form in $(y_1)$. Since $\alpha$ is general and $\beta$ can also be chosen to be general, we have $F\in (W)$ in $R$ by part (4) of \cref{proposition: general quotient general}.

    \emph{Case 2.} Now suppose that we have $G\not \in (W_\alpha)$ and $G$ is absolutely reducible over $W_\alpha$ for all $\alpha\in \cZ$. By \cref{proposition: general quotient general} there exists a non-empty open subset $\cV\subseteq \bK^{n+1}$ such that $\varphi_{Y_\alpha,\gamma\cdot(1,\beta)}(G)$ is reducible for all $\gamma\cdot(1,\beta)\in \cV$. We let $\cU_\alpha$ be the image of $f^{-1}(\cV)$ under the projection map $p_2:\bK\times \bK^{\dim(W')}\rightarrow \bK^{\dim(W')}$, where $f(\gamma,\beta)=\gamma\cdot(1,\beta)$ as in \cref{proposition: general twists}. 
    
    Note that $\varphi_{W,(\alpha,\beta)}=\varphi_{Y_\alpha,\gamma\cdot(1,\beta)}\circ \varphi_{V,\alpha}$ by \cref{prop: composition of quotient}. Let  $\cZ':=\{(\alpha,\beta)\mid \alpha\in \cZ, \beta\in \cU_\alpha\}$. Then $\cZ'\subseteq \bK^{\dim(W)}$ is a dense set. Now for all $(\alpha,\beta)\in \cZ'$, we have $\varphi_{W,(\alpha,\beta)}(F)$ is reducible in $R[y_1]/I_{W,(\alpha,\beta)}$, by the isomorphism in \cref{prop: composition of quotient}. As $F\not \in (W)$, we conclude that $F$ is absolutely reducible over $W$ by \cref{proposition: general quotient general}. Since $F\in \cA^V_d$ was arbitrary, we conclude that $m_a^W=0$ as desired.
    The claimed bound on the strength of $W$ follows from \cref{lem: iterated lifted strength}.
\end{proof}

\begin{lemma} \label{lem: EK low degree control}
    Suppose $\abs{\cC} \geq 2 \primebound{d} \varepsilon^{-1}$.
    Let $k := 10 \cdot d \cdot \br{\irredprimebound{d} + \primebound{d}} \cdot \br{\irredprimebound{d} + 1}$.
    Let $V$ be a $H(\strong, k, 8 \cdot d \cdot \varepsilon^{-1})$-lifted strong vector space such that one of the following holds.
    \begin{itemize}
        \item  $M_{a}^{V} \geq 6 \varepsilon \abs{\cA}$ or
        \item $M_{b}^{V} \geq 6 \varepsilon \abs{\cB}$ or
        \item $M_{c}^{V} \geq 6 \varepsilon \abs{\cC}$ and $\abs{\cC} \geq \abs{\cB} / 2$.
    \end{itemize}
    Then there is a $\strong$-lifted strong vector space $W$ obtained by applying a $(k, 8 \cdot d \cdot \varepsilon^{-1})$ process to $V$ such that $m_{a}^{W}\cdot m_{b}^{W}\cdot m_{c}^{W}=0$.
\end{lemma}

\begin{proof}
    We first assume $M_{b}^{V} \geq 6 \varepsilon \abs{\cB}$.
    In this case a $(k, 4 \cdot d \cdot \varepsilon^{-1})$ process will suffice.
    Recall that for any vector space $W$, the set $\cA_d^W$ is the set of forms in $\cA_d$ which are not in $(W)$ and are absolutely irreducible with respect to $W$. Also, we have $m_a^W:=|\cA_d^W|$ and $M_a^W:=|\cA|-m_a^W$.

    Let $\sumprimebound{d} := \irredprimebound{d} + \primebound{d}$.
    Let $a := 10 \cdot d \cdot \sumprimebound{d}$.
    We will describe one iteration of the desired $(k, 4 \cdot d \cdot \varepsilon^{-1})$ process below and we will show that at most $4 \cdot d \cdot \varepsilon^{-1}$ iterations suffice. 

    First let us note two special cases which yield a vector space $W$ with $m_{a}^{W}\cdot m_{b}^{W}\cdot m_{c}^{W}=0$. If we reach any of these two \emph{breaking conditions} during the iterative process, we will terminate the process. The first condition will involve properties of $\cB$ and the second will involve $\cA$.

    \emph{Breaking condition $1$.} Suppose that $\abs{\cB \cap \bK\bs{V}} \geq \varepsilon \abs{\cB}$. Then we apply \cref{lem: ek algebra control} to $V$, and call the resulting space $W$.
    We terminate the iterative process. 
    
    \emph{Breaking condition $2$.} Suppose that there exists a dense subset $\cZ\subseteq \bK^{\dim(V)}$ such that for all $\alpha\in \cZ$ we have

    \[a> \max\{|\cS|\mid \cS\subseteq \cA_d^V \text{ and } \varphi_{V,\alpha}(\cS)\text{ is an unbreakable pencil generating set}\}.\]
    Here $\varphi_{V,\alpha}:R[y_1]\rightarrow R[y_1]/I_{V,\alpha}$ is a graded quotient. Then we apply \cref{lemma: breaking condition 2} and let $W$ be the resulting vector space. We terminate the iterative process.
    
   \textit{ Iterative step.} Suppose that none of the breaking conditions occur. Then we have $|\cB\cap\bK[V]|<\varepsilon |\cB|$. Moreover, there exists a non-empty open subset $\cU\subseteq \bK^{\dim(V)}$, such that for all $\alpha\in \cU$ there exists a set $\cS_\alpha\subseteq \cA^V_d$ where $|\cS_\alpha|=a$ and $\varphi_{V,\alpha}(\cS_\alpha)$ is an unbreakable pencil generating set. Since there are finitely many possibilities for such $\cS_\alpha$, we may replace $\cU$ by a dense subset $\cZ\subseteq\bK^{\dim(V)}$ and assume that all the $\cS_\alpha$ are the same set for all $\alpha\in \cZ$. Let $\cS=\cS_\alpha$ for all $\alpha$. Note that $\cS\not\subseteq \bK[V]$, as $\cS\subseteq \cA^V_d$. Therefore, we may apply \cref{proposition: gcd after projection} and assume that $|\varphi_{V,\alpha}(\cS)|=|\cS|$, after possibly replacing $\cZ$ with a smaller dense set. Moreover, we may also assume that all elements of $\varphi_{V,\alpha}(\cS)$ are irreducible of degree $d$. We update $V$ to $AH(V, \cS)$ and go to the next iteration.

    Now we will show that after at most $2 \cdot d \cdot \varepsilon^{-1}$ iterations of the above process, one of the two break conditions must be met.
    This implies that the above process is a $(k, 4 \cdot d \cdot \varepsilon^{-1})$ process as required, and that the required conclusion holds.

    \emph{Analysis of one iteration.} Suppose that $|\cB\cap\bK[V]|<\varepsilon |\cB|$. Moreover, there exists a set $\cS\subseteq \cA_d^V$ with $|\cS|=a$ and a dense set $\cZ\subseteq \bK^{\dim(V)}$, such that for all $\alpha\in \cZ$ the set $\varphi_{V,\alpha}(\cS)$ is an unbreakable pencil generating in $R[y_1]/I_{V,\alpha}$. Let $R'_\alpha:=R[y_1]/I_{V,\alpha}$. We define $Y:=AH_R(V,\cS)$. Note that $Y$ is a $h_{2\mu} \circ t_{1}$-lifted strong vector space in $R'_\alpha$ by \cref{lem: iterated lifted strength}, since it is a vector space that is intermediate in a $(k, 4 \cdot d \cdot \varepsilon^{-1})$ process and since $V$ is strong enough. Moreover we have $V\subseteq Y$. By choosing compatible bases of $V\subseteq Y$, we let $Y=V+Y'$ and $Y'_\alpha:=\varphi_{V,\alpha}(Y)\subseteq R'_\alpha$, as in \cref{lemma: abs red and potential increase}. Let $Y_\alpha=\Kspan{y_1,Y'_\alpha}.$

    Let $B_{1}, \dots, B_{r} \in\cB$ be the forms that either have degree less than $d$, or that are absolutely reducible with respect to $V$ or in the ideal $(V)$.
    By definition, we have $r = M_{b}^{V}$.
    Among these forms, let $B_{1}, \dots, B_{s}$ be the set of forms that are not in $\bK\bs{V}$.
    Since $\abs{\cB \cap \bK\bs{V}} < \varepsilon \abs{\cB}$, we have $s \geq M_{b}^{V} - \varepsilon \abs{\cB}$.

    For $\alpha\in \cZ$, let $\cG_\alpha\subseteq \{B_1,\cdots,B_r\}$ be the (possibly empty) subset of forms
    such that some factor of $\varphi_{V,\alpha}(B_i)$, other than $y_1$, is either absolutely reducible over $Y_\alpha$, or in the ideal $(Y_\alpha)$. Now there are finitely many possibilities for $\cG_\alpha$ as $B_1,\cdots,B_r$ are fixed. Therefore, by replacing $\cZ$ with a possibly smaller dense subset, we may assume that there exists a set $\cG$ such that $\cG_\alpha=\cG$ for all $\alpha\in \cZ$. In particular, if $B_i\in \cG$ then some factor of $\varphi_{V,\alpha}(B_i)$, other than $y_1$, is either absolutely reducible over $Y_\alpha$, or in the ideal $(Y_\alpha)$, for all $\alpha\in \cZ$.
    Therefore, if $B_i \in \cG$ then $\Psi_{Y,(\alpha,\beta)}(B_i) > \Psi_{V,\alpha}(B_i)$, for all $\alpha\in \cZ$ and for general $\beta\in \bK^{\dim(Y')}$ (depending on $\alpha$) by \cref{lemma: abs red and potential increase}. We will show below that $\abs{\cG} \geq \varepsilon \abs{\cB}$. 
    
    First, let us show that this condition  implies that there can be at most $2 \cdot d \cdot \varepsilon^{-1}$ number of iterations before a breaking condition is reached.

    \emph{Bounding total number of iterations.} Suppose we know that $\abs{\cG} \geq \varepsilon \abs{\cB}$ at each iteration. Now each form in $\cG$ is potential increasing with respect to $(V,\alpha)$ and $(Y,(\alpha,\beta))$, for all $\alpha\in \cZ$ and general choice of $\beta$ (depending on each $\alpha$). Then, in each iteration, when we increase $V$ to $Y$, the group potential of $\cB$ increases by at least $\varepsilon \abs{\cB}$, i.e.
    \[\Phi_{Y,(\alpha,\beta)}(\cB) > \Phi_{V,\alpha}(\cB)\]
    for a dense set of $\alpha\in \cZ'$ and general $\beta$ (depending on $\alpha$). In other words, the group potential jumps by at least $\varepsilon|\cB|$ for a dense set of $(\alpha,\beta)$, when we replace $V$ by $Y$.
    Note that the group potential of $\cB$ is upper bounded by $2 d \abs{\cB}$ by \cref{remark: group potential bound}, irrespective of the choice of $\alpha,\beta$.  Therefore there can be at most $2 \cdot d \cdot \varepsilon^{-1}$ iterations that do not reach a break condition. 

    \emph{Showing $|\cG|\geq \varepsilon|\cB|$ for each iteration.} Suppose that we have fixed $V$ and a dense set of $\cZ\subseteq \bK^{\dim(V)}$ as above. Suppose that $|\cG|<\varepsilon |\cB|$. We can further reorder $B_{1}, \dots, B_{s}$ so that the elements of $\cG$ are at the end of the ordering.
    Suppose $B_{1}, \dots, B_{t} \not \in \cG$, we have $t \geq M_{b}^{V} - 2 \varepsilon \abs{\cB} \geq 4 \varepsilon \abs{\cB}$.

    For each $B_{i}$ with $i \leq t$, fix an irreducible factor $B'_{i}$ of $\varphi_{V,\alpha}(B_{i})$, such that $B_i'$ is not a scalar multiple of $y_1$. Note that $\deg(B_j')<d$, by the choice of $B_j$ and \cref{proposition: general quotient general}. Since $B_i\not \in \cG$, we know that $B_i'$ is absolutely irreducible over $Y_\alpha$ and $B_i'\not \in (Y_\alpha)$. Recall that $\cS\subseteq \bK[Y]$, and hence $\varphi_{V,\alpha}(\cS)\subseteq \bK[Y_\alpha]$. Moreover, $Y_\alpha$ is sufficiently strong by \cref{prop: composition of quotient}, so that we may apply \cref{cor:strongprimebound}. Hence each $B'_{i}$ is not prime with at most $\primebound{d}$ elements in  $\varphi_{V,\alpha}(\cS)$. Moreover each $B_i'$ spans something reducible with at most $\irredprimebound{d}$ many elements in $\varphi_{V,\alpha}(\cS)$. 
    
    Recall that $|\varphi_{V,\alpha}(\cS)|=|\cS|=a$. Let $\varphi_{V,\alpha}(\cS)=\{F_1,\cdots,F_a\}$. We denoted $\sumprimebound{d}=\primebound{d}+\irredprimebound{d}$. Then by a double counting, at most $2 \sumprimebound{d}$ of the forms among $F_{1}, \dots, F_{a}$ are not prime with more than $t/2$ of the $B'_{i}$.
    Let $b := a - 2 \sumprimebound{d}$. We may reorder and assume that for any fixed $i\in [b]$, there are at least $t/2$ forms among the $B'_{j}$, such that $(F_i,B_j')$ is prime.

    By \cref{prop: graded quot of ek} we know that the sets $\varphi
    _{V,\alpha}(\cA),\varphi
    _{V,\alpha}(\cB),\varphi
    _{V,\alpha}(\cC)$ form a $(d,z,R)$-EK configuration. Therefore, for each $F_i,B'_j$ which form a prime ideal, there exists a form $C\in \cC$ such that $\varphi_{V,\alpha}(C) \in \ideal{F_{i}, B'_{j}}$. For $i\in [b]$, we let $\cP_{i,\alpha}$ be the set of forms $C \in \cC_{d}$ such that $\varphi_{V,\alpha}(C) \in \ideal{F_{i}, B'_{j}}$ for some $j$ such that $\ideal{F_i, B_{j}'}$ is prime. Now there are finitely many possibilities for the sets $\cP_{i,\alpha}\subseteq \cC_d$. Therefore, by replacing $\cZ$ with a possibly smaller dense subset, we may assume that there exists a set $\cP_i$ such that $\cP_i=\cP_{i,\alpha}$ for all $\alpha\in \cZ$.
    Recall that $F_{i}$ is of degree $d$. Moreover $B_j'$ does not divide $\varphi_{V,\alpha}(C)$ by \cref{proposition: gcd after projection}, since $B_j'\not \in (y_1)$ and $\alpha$ varies in a dense set. Therefore $\varphi_{V,\alpha}(C)$ is irreducible and has degree $d$, as $\varphi_{V,\alpha}(C)\in (F_i,B_j')$. Since this holds for all $\alpha$ in a dense set, we see that $C$ must be absolutely irreducible over $V$, by \cref{proposition: general quotient general}.

    We first prove that each $|\cP_{i}|\geq t/2(d+1) \geq t / 4d$.
    Fix $i = 1$ without loss of generality.
    Suppose $F_{1}$ is prime with $B'_{1}, \dots, B'_{t/2}$.
    Suppose $\abs{\cP_{1}} < t / 2(d+1)$. Then there exists $C\in \cP_i$ such that $\varphi_{V,\alpha}(C)\in (F_1,B_j')$ for all $j\in [d+1]$, after a possible reordering of $B'_{1}, \dots, B'_{t/2}$. Therefore we have \[c_{1} F_{1} + B'_{1} H_{1} = \cdots = c_{d+1} F_{1} + B'_{d+1} H_{d+1}\]
    for some scalars $c_i\in \bK$ and forms $H_i\in R'_\alpha$. Recall that and $\deg(B_i')<d$ for all $i$ and $\deg(F_1)=d$.
    If $c_{i} \neq c_{j}$ for some $i \neq j$, then we have $\ideal{F_{1}, B'_{i}} \subseteq \ideal{B'_{i}, B'_{j}}$. This is a contradiction since $\ideal{F_{1}, B'_{i}}$ is prime.
    Therefore we have $c_{1} = \cdots = c_{d+1}$.
    Now each of $B'_{2}, \dots, B'_{d+1}$ are factors of $H_{1}$, which is a contradiction since $H_{1}$ has at most $d$ factors.

    Since $t\geq 4\varepsilon |\cB|$, we have $|\cP_{i}|\geq \varepsilon \abs{\cB} / d \geq \varepsilon m_{c}^{V} / d$.
    The number of sets $\cP_{i}$ is $b := a - 2 \sumprimebound{d}$.
    Hence $b > d (\sumprimebound{d} + 2) \varepsilon^{-1}$. By double counting, we can assume that there is a common element in $\cP_{1}, \dots, \cP_{\sumprimebound{d} + 2}$.
    We use this to derive the contradiction, completing the proof.
    
    Suppose $C$ lies in $\cP_{1}, \dots, \cP_{\primebound{d} + 2}$ with $C' = \varphi_{V,\alpha}(C)$.
    Therefore, for each $F_i$ there exists $B_i'$ such that $C' = c_{i} F_{i} + H_{i} B'_{i}$.
    We therefore have $\ideal{c_{1} F_{1} - c_{i} F_{i}, B'_{1}} \subseteq \ideal{B'_{1}, B'_{i}}$ for all $2 \leq i \leq \sumprimebound{d} + 1$.
    Recall that $\varphi_{V,\alpha}(\cS)$ is an unbreakable pencil generating set. Hence the forms $c_{1} F_{1} - c_{i} F_{i}$ are nonzero and pairwise non associate, by the linear independence of $F_{1}, \dots, F_{a}$.
    Further, each element is irreducible as $\smin{F_i,F_j}\geq 1$.
    Recall that $B_1'$ is absolutely irreducible over $Y_\alpha$ and $B_1'\not \in (Y_\alpha)$. Therefore, by \cref{cor:strongprimebound}, there is a $j$ such that $\ideal{c_{1} F_{1} - c_{j} F_{j}, B'_{1}}$ is prime. Since $B'_{1}, B'_{j}$ are irreducible in $R'_\alpha$, which is a UFD, the ideal has height $2$. Therefore we must have an equality $\ideal{c_{1} F_{1} - c_{j} F_{j}, B'_{1}} = \ideal{B'_{1}, B'_{j}}$. This is a contradiction.

    In the case when $M_{a}^{V} \geq 6 \varepsilon \abs{\cA}$ the exact argument holds with all occurrences of $\cA, \cB$ swapped.
    Now suppose we are in the case when $M_{c}^{V} \geq 6 \varepsilon \abs{\cC}$ and $\abs{\cC} \geq \abs{\cB} / 2$.
    We swap the roles of $\cC$ and $\cB$.
    Note that in this case, $\abs{\cC \cap \bK\bs{V}} \geq \varepsilon \abs{\cC}$ is also a valid breaking condition, since \cref{lem: ek algebra control} can be applied with $\cC$ by assumption.
    The only other place in the arguments where the relative sizes of $\cA, \cB, \cC$ play a role is to ensure that the sets $\cP_{i}$ have a common intersection.
    Each set in this case will have size $\varepsilon \abs{\cC} / d$, which is at least $\varepsilon m_{b}^{V} / 2d$ by assumption.
    The rest of the argument remains the  same.
\end{proof}

\begin{lemma} \label{lem:EK control one to all}
    Suppose $V$ is a $H(\strong, \primebound{d} + \irredprimebound{d} + 1, 2)$-lifted strong space such that one of $m_{a}^{V}, m_{b}^{V}, m_{c}^{V}$ is zero.
    Then there is a $\strong$-lifted strong space $W$ obtained by applying a $(\primebound{d} + \irredprimebound{d} + 1, 2)$ process to $V$ such that $m_{a}^{W} = m_{b}^{W} = m_{c}^{W} = 0$.
\end{lemma}

\begin{proof}
    In this proof, the roles of $\cA, \cB, \cC$ can be freely permuted.
    Therefore without loss of generality assume $m_{a}^{V} = 0$.
    Let $\sumprimebound{d} := \primebound{d} + \irredprimebound{d}$.

    Let $B_{1}, \dots, B_{r}$ be the forms in $\cB_{d}$ that are absolutely irreducible over $V$.
    If the span of $B_{1}, \dots, B_{r}$ is less than $\sumprimebound{d} + 1$, then we add a basis of this set to $V$, and call this resulting space $Y$.
    We have $m_{b}^{Y} = m_{a}^{Y} = 0$, and we move to the next step.
    Suppose the span is larger than $\sumprimebound{d} + 1$.
    Without loss of generality, assume $B_{1}, \dots, B_{\sumprimebound{d}+ 1}$ are linearly independent.
    Let $Y$ be the space $AH(V, B_{1}, \dots, B_{\sumprimebound{d} + 1})$.
    Now let $C$ be a form in $\cC_{d}$ that is absolutely irreducible over $Y$, if such a form exists.
    There is at least one form among $B_{1}, \dots, B_{\sumprimebound{d} + 1}$ say $B_{i}$ such that $\ideal{C, B_{i}}$ is prime, and does not span anything absolutely reducible.
    However, the EK image of $B_{i}, C$ has to be a form in $\cA_{d}$ that is absolutely irreducible over $V$, but no such form exists.
    Therefore, we have $m_{c}^{Y} = 0$.
    This completes the first step of our process.

    After rearranging again, we can assume now that $m_{a}^{Y} = m_{b}^{Y} = 0$.
    Among the forms in $\cA_{d}, \cB_{d}$, pick a set of forms $F_{1}, \dots, F_{\sumprimebound{d} + 1}$ that are linearly independent.
    If such forms exist, set $W = AH(Y, F_{1}, \dots, F_{\sumprimebound{d} + 1})$.
    Each form $C \in \cC_d$ that is absolutely irreducible over $W$ has at least one $F_{i}$ such that $\ideal{F_{i}, C}$ is prime and does not span anything absolutely reducible, but the EK image of this ideal is in $\cA_{d} \cup \cB_{d}$ which only consists of absolutely reducible forms.
    
    Suppose $F_{1}, \dots, F_{\sumprimebound{d} + 1}$ cannot be picked.
    Then we instead define $F_{1}, \dots, F_{\sumprimebound{d} + 1}$ to be a set of linearly independent forms in $\cA \cup \cB$ that consist of a basis of $\cA_{d} \cup \cB_{d}$ (the forms are now allowed to be of lower degree), if possible.
    Again let $W = AH(Y, F_{1}, \dots, F_{\sumprimebound{d} + 1})$.
    Again if $C \in \cC_{d}$ is absolutely irreducible over $W$ then we have $\ideal{C, F_{i}}$ is prime for some $i$, and therefore $G = \alpha C + F_{i} H_{i}$ for some degree $d$ form $G \in \cA_{d} \cup \cB_{d}$.
    This implies that $C$ is in the ideal generated by $W$, since by assumption every form in $\cA_{d} \cup \cB_{d}$ is in the span of $W$.
    This is again a contradiction.

    The only remaining case is when $\cA \cup \cB$ has span at most $\sumprimebound{d} + 1$.
    In this case, we just pick $F_{1}, \dots, F_{a}$ to be a basis of $\cA \cup \cB$, and set $W = AH(V, F_{1}, \dots, F_{a})$.
    Every element of $\cA \cup \cB$ is in the algebra $\bK\bs{W}$.
    Now for any $C \in \cC_{d}$ that is not in the ideal, and for any element say $A \in \cA$, the radical ideal $\radideal{A, C}$ has to contain $\prod_{B \in \cB} B$.
    The latter element is in $\bK\bs{W}$, and by the elimination theorem \cite[Lemma~4.26]{OS24}, the only element in $\radideal{A, C} \cap \bK\bs{W}$ is $A$, which is a contradiction.

    Therefore in all cases, we have $m_{a}^{W} = m_{b}^{W} = m_{c}^{W} = 0$.
    The claimed bound on the strength of $W$ follows from \cref{lem: iterated lifted strength}.
\end{proof}

The following lemma combines the above results and allows us to fully control the highest degree forms in an EK-configuration.

\begin{lemma}
    \label{lem:induction step main}
    Let $c := 3 (\primebound{d} + \irredprimebound{d} + 1)$ and $\delta := 8 \cdot \varepsilon \cdot (\primebound{d} + \irredprimebound{d} + 1)^{-1}$.
    Let $k := 6 \cdot d \cdot (\primebound{d} + \irredprimebound{d} + 1)^{3} \cdot \varepsilon^{-1} + c + c_{\mathrm{ek}} \cdot \delta^{-1}$ and $t := 8 \cdot d \cdot \varepsilon^{-1} + 3$.
    Let $\strong$ be a function, and suppose $U$ is $H(\strong, k, t)$-strong.
    There is a $(k, t)$ process starting with $0$ such that the resulting space $W \subset R$ is $\strong$-lifted strong and satisfies $m_{a}^{W} = 0$ and $m_{b}^{W} = 0$ and $m_{c}^{W} = 0$.
    Further, $\dim W_{i} \leq D_{i}(\strong, k, t, \delta_{U})$, where $\delta_{U}$ is the dimension vector of $U$.
\end{lemma}

\begin{proof}
    The first step is to find a space $V$ such that one of $m_{a}^{V}, m_{b}^{V}, m_{c}^{V}$ is zero.
    
    If $\abs{\cC} < 2 \primebound{d} \varepsilon^{-1}$ then we pick $V = AH(0, \cC)$ and we are done.
    If $M_{b}^{(0)} \geq 6 \varepsilon \abs{\cB}$ or $M_{a}^{(0)} \geq 6 \varepsilon \abs{\cA}$ then we apply \cref{lem: EK low degree control} starting with $0$ and we are done.

    Now consider the sets $\cA_{d}, \cB_{d}, \cC_{d}$.
    Even though we have $\abs{\cA} \geq \abs{\cB} \geq \abs{\cC}$, we cannot assume $\abs{\cA_{d}} \geq \abs{\cB_{d}} \geq \abs{\cC_{d}}$.
    For each form $F \in \cF_{d}$ we define a subset $\Fspan[F] $ as follows.
    If $F \in \cA_{d}$ and $\abs{\cB_{d}} \geq \abs{\cC_{d}}$ then $\Fspan[F] \subset \cB_{d}$ is the set of forms $B \in \cB_{d}$ such that $\abs{\Kspan{F, B} \cap \cC_{d}} \geq 1$.
    If instead $\abs{\cC_{d}} \geq \abs{\cB_{d}}$ then $\Fspan[F] \subset \cC_{d}$ is the set of forms $C \in \cC_{d}$ such that $\abs{\Kspan{F, C} \cap \cB_{d}} \geq 1$.
    We can similarly define $\Fspan[F]$ for forms $F \in \cB_{d}, \cC_{d}$.
    In each case, $\Fspan[F]$ will be a subset of the bigger of the two degree $d$ parts of the remaining sets.
    The fractional size of $\Fspan[F]$ is defined to be the ratio of $\abs{\Fspan[F]}$ and the size of whichever among $\cA_{d}, \cB_{d}, \cC_{d}$ that $\Fspan[F]$ is a subset of.

    For a form $F \in \cA_{d} \cup \cB_{d}$, we have $\Fspan[F] \subset \cC_{d}$ only if $\abs{\cC_{d}} \geq \abs{\cB_{d}}$ or $\abs{\cC_{d}} \geq \abs{\cA_{d}}$.
    We are in the case where $M_{b}^{(0)} < 6 \varepsilon \abs{\cB} < \abs{\cB} / 2$ and $M_{a}^{(0)} < 6 \varepsilon \abs{\cA} < \abs{\cA} / 2$, equivalently in the case where $\abs{\cB_{d}} \geq \abs{\cB} / 2$ and $\abs{\cA_{d}} \geq \abs{\cA} / 2$.
    Therefore if $\Fspan[F] \subset \cC_{d}$ then we have $\abs{\cC} \geq \abs{\cB} / 2$.
    This observation will be useful later in the proof.
    
    If the set $\cA_{d} \cup \cB_{d} \cup \cC_{d}$ is a $(c, \delta)$-linear EK configuration (as defined in \cref{def:linearek}), then by \cref{prop:linearek} we can find a basis of $\cA_{d} \cup \cB_{d} \cup \cC_{d}$ of size at most $c + c_{\mathrm{ek}} \cdot \delta^{-1}$, and we are done.
    Suppose they do not form such a configuration.
    In particular, there are at least $c$ forms in $\cF_{d}$ with fractional $\Fspan[F]$ smaller than $\delta$.
    By the pigeonhole principle, there are at least $c/3 = \primebound{d} + \irredprimebound{d} + 1$ such forms within the same set.
    We now do a case analysis.

    \paragraph{Case $\abs{\cC} \leq \abs{\cB} / 2$:}

    The key observation in this case is that none of the $\Fspan[F]$ are subsets of $\cC_{d}$.
    We analyse the subcase when there are $c/3$ forms in $\cC_{d}$ with $\Fspan[C_{i}] \subset \cA_{d}$ having fractional size less than $\delta$.
    The other subcases have the exact same proof (with the roles of $\cA, \cB$ swapped if $\abs{\cA_{d}} \leq \abs{\cB_{d}}$).
    Suppose forms $C_{1}, \dots, C_{c/3} \in \cC_{d}$ have fractional size less than $\delta$, and suppose $\Fspan[C_{i}] \subset \cA_{d}$.
    Let $Y := AH(0, C_{1}, \dots, C_{c/3})$.
    If there is some $A \in \cA_{d}$ such that $\ideal{C_{i}, A}$ is prime, then $A \in \Fspan[C_{i}]$.
    Therefore each $C_{i}$ is prime with at most $\delta \abs{\cA_{d}}$ elements of $\cA_{d}$.
    By a union bound, there are are at most $(c/3) \delta \abs{\cA_{d}} \leq 4 \varepsilon \abs{\cA}$ elements in $\cA_{d}$ that are prime with some element among $C_{1}, \dots, C_{c/3}$.
    The remaining $(1 - 10 \varepsilon) \abs{\cA}$ elements of $\cA$ are either absolutely reducible over $Y$ or in the algebra.
    Therefore, we have $M_{a}^{Y} \geq (1 - 10 \varepsilon) \abs{\cA} \geq 6 \varepsilon \abs{\cA}$.
    Now we can apply \cref{lem: EK low degree control} starting with $Y$ and the resulting space $V$ is such that one of $m_{a}^{V}, m_{b}^{V}, m_{c}^{V}$ is zero.

    \paragraph{Case $\abs{\cC} \geq \abs{\cB} / 2$:}

    In this case, it is possible that some forms $F \in \cF_{d}$ are such that $\Fspan[F] \subset \cC_{d}$.
    If we are in this case and we have $M_{c}^{(0)} \geq 6 \varepsilon \abs{\cC}$ then we can \cref{lem: EK low degree control} starting with $0$ and we are done.
    Therefore we have to deal with the subcase when $M_{c}^{(0)} \leq 6 \varepsilon \abs{\cC}$.

    If there are $c/3$ forms $F_{1}, \dots, F_{c/3}$ all in $\cA_{d}$ or $\cB_{d}$ with $\Fspan[F] \subset \cC_{d}$ and fractional size less than $\delta$, then we can set $Y := AH(0, F_{1}, \dots, F_{c/3})$, and we have $M_{c}^{Y} \geq (1 - 10 \varepsilon) \abs{\cC} \geq 6 \varepsilon \abs{\cC}$.
    Therefore we can invoke \cref{lem: EK low degree control}.
    In the other cases, for example if the $c/3$ forms are in $\cC_{d}$, then the analysis from the previous case applies.
    Therefore, we have $V$ such that one of $m_{a}^{V}, m_{b}^{V}, m_{c}^{V}$ is zero.

    Now that we have the space $V$, we apply the process of \cref{lem:EK control one to all} to obtain $W$.
    The fact that this entire process is a $(k, t)$ process for the claimed $k, t$ follows simply by adding up the number of iterations in the intermediate processes, and taking a max for the number of forms added at each step of the intermediate processes (for incomparable terms we use the sum to bound the max).
    The claimed bounds on the strength and dimension of $W$ follow from \cref{lem: iterated lifted strength} and the observation that since $U$ is $H(\strong, k, t)$-strong, the $0$ vector space in $R$ is $H(\strong, k, t)$-lifted strong.
\end{proof}

\subsection{Putting it all together}

We can now prove our main theorem for EK-configurations, which we restate here for convenience.

\generalekmain*

\begin{proof}
    For each fixed $e$, we prove the result by induction on $d$.
    Therefore, for the rest of this proof we fix $e$.
    Fix $\strong: \bN^{e} \to \bN^{e}$ such that $\strong_{i}(\delta) = A(\eta, i) + 3 \norm{\delta}_{1}$, where $A$ is the function defined in \cite[Theorem~A]{AH20}.
    We will define the functions $\Lambda_{d, e}$ and $\lambda_{d, e}$ inductively, and simultaneously prove that they have the claimed properties.

    The base case is when $d = 1$.
    Define $\Lambda_{1, e} := h_{2 \strong} \circ t_{2}$.
    If $U$ is $\Lambda_{1, e}$-strong for this choice, then any two linear forms in $R$ that are non associate form a prime ideal by \cite[Proposition~5.10]{OS24}.
    In particular, for any $A \in \cA, B \in \cB$ the ideal $\ideal{A, B}$ is prime, and $\radideal{A, B} = \ideal{A, B}$.
    The second condition for $(1, z, R)$-EK configurations reduces to the condition that the linear span of $A, B$ contains some element in $\cC \cup \bc{z}$.
    The third and fourth conditions reduce to similar symmetric statements.
    This shows that the set $\cF$ is a $(1, 1)$-partial linear EK configuration, where the error subset $\cG$ is the singleton $\bc{z}$, and the partitions $\cA, \cB, \cC$ witness the required partition in the definition of $(1, 1)$-partial linear EK configurations.
    Therefore we can bound the dimension of $\cF$ by $\lambda_{1, e} := 1 + c_{\mathrm{ek}}$.
    This completes the base case.

    Inductively, suppose we have proved the theorem for $d-1$, and in particular suppose the functions $\Lambda_{d-1, e}, \lambda_{d-1, e}$ are defined.
    Let $c := 3 (\primebound{d} + \irredprimebound{d} + 1)$ and $\delta := 8 \cdot \varepsilon \cdot (\primebound{d} + \irredprimebound{d} + 1)^{-1}$.
    Let $k := 6 \cdot d \cdot (\primebound{d} + \irredprimebound{d} + 1)^{3} \cdot \varepsilon^{-1} + c + c_{\mathrm{ek}} \cdot \delta^{-1} \cdot \log{\delta^{-1}}$ and $t := 8 \cdot d \cdot \varepsilon^{-1} + 3$.
    Set $\Lambda_{d, e} := H(\Lambda_{d-1, e}, k, t)$.

    Suppose $U$ is $\Lambda_{d, e}$-strong, and $\cA, \cB, \cC$ is an $(d, R, z)$-EK configuration.
    By \cref{lem:induction step main}, there exists a vector space $W \subset R$ that is $\Lambda_{d-1, e}$-lifted strong such that every degree $d$ form in $\cF$ is either in the ideal generated by $W$, or absolutely reducible with respect to $W$.
    For a general choice of $\alpha$, consider the graded quotient $R\bs{y} \to R\bs{y} / I_{\alpha}$ where $I_{\alpha}$ is the ideal is the ideal corresponding to $W$ and $\alpha$ in $R\bs{y}$.
    By \cref{prop: graded quot of ek}, the image of the configuration $\cA, \cB, \cC$ is itself an $(d, y, R\bs{y} / I_{\alpha})$-EK configuration.
    Further, since every degree $d$ form in $\cF$ is either in the ideal generated by $W$, or absolutely reducible with respect to $W$, we can deduce by \cref{proposition: general quotient general} that the image is in fact a $(d-1, y, R\bs{y} / I_{\alpha})$-EK configuration.
    Note that $R\bs{y} / I_{\alpha} = S\bs{y} / U + I_{\alpha}$.
    Further, since $W$ is $\Lambda_{d-1, e}$-lifted strong, we can deduce that $U + I_{\alpha}$ is generated by a $\Lambda_{d-1, e}$-strong vector space of dimension at most $\dim U + \dim W$.
    We can apply the inductive hypothesis to deduce that the image of the configuration under the graded quotient has dimension at most $\lambda_{d-1, e}(\dim W + \dim U)$.
    By \cref{proposition: lifting general quotient} we can deduce that the dimension of the original configuration is bounded by $\lambda_{d-1, e}(\dim W + \dim U) \cdot (1 + d)^{3 \dim W + 4} \cdot e^{\dim U}$.
    We also have $\dim W \leq \sum_{i=1}^{d} D_{i}(\Lambda_{d-1, e}, k, t, \delta_{U})$.
    Therefore, if we set
    $$\lambda_{d, e}(u) := \max_{\substack{\delta \in \bN^{e} \\ \norm{\delta}_{1} = u}} \lambda_{d-1, e}\br{\sum_{i=1}^{d} D_{i}(\Lambda_{d-1, e}, k, t, \delta) + u} \cdot (1 + d)^{3 \sum_{i=1}^{d} D_{i}(\Lambda_{d-1, e}, k, t, \delta) + 4} \cdot e^{u},$$
    then we are done.
\end{proof}

A straightforward corollary is our main theorem on EK-configurations.

\ekmain*

\begin{proof}
    The vector space $(0)$ is $\strong$-strong for any function $\strong$.
    Therefore we can invoke \cref{theorem: general ek main} with $U = (0)$, and it suffices to pick $\lambda(d) := \lambda_{d, d}(0)$.
\end{proof}

\section{Rank bounds and PIT for depth four identities} \label{sec: circuits PIT}

We start this section by formally establishing the relationship between EK-configurations and $\tspsp$ circuits.
In particular, we show how to define an EK configuration from a $\tspsp$ circuit, and how the ranks of the two objects are related.
Our approach here is to create a configuration by collecting together all irreducible factors of the polynomials computed at each gate, and then homogenising them.

\begin{definition}[Configuration corresponding to $\tspsp$ circuits.]\label{def:circuit to config}
    Given a $\tspsp$ circuit $T = \sum_{i=1}^{3} \prod_{j} P_{ij}$, we define a configuration of forms as follows.
    Let $\cA'$ be the set of irreducible factors of the polynomials $\bc{P_{1j}}_{j}$.
    Define $\cA$ to be the set of forms obtained by homogenising the forms in $\cA'$, where we homogenise using a new variable $z$.
    If $\cA$ contains a subset of forms that are pairwise associate, then we discard all except one form from this subset from $\cA$, so the final set $\cA$ does not have any pair of forms that are associate to each other.
    Similarly, the sets $\cB, \cC$ are defined using the polynomials $\bc{P_{2j}}_{j}$ and $\bc{P_{3j}}_{j}$ respectively.

    The sets $\cA, \cB, \cC$ consist of irreducible homogeneous forms of degree at most $d$ from the ring $S\bs{z}$.
    Further, neither of the three sets contains the form $z$.
\end{definition}

\begin{lemma}
    \label{lem:tspsps is EK}
    Suppose $T = \sum_{i=1}^{3} \prod_{j} P_{ij} = \sum_{i=1}^{3} T_{i}$ is a simple minimal circuit that computes $0$.
    The sets $\cA, \cB, \cC$ as defined in \cref{def:circuit to config} are a $(d, z, S\bs{z})$-EK configuration.
    Further, $\Rank{T} \leq \lambda^{d}$ where $\lambda$ is the rank of the EK configuration.
\end{lemma}

\begin{proof}
    Note that $T$ continues to compute zero, and continues to be simple and minimal when we consider it as a polynomial in $S\bs{z}$.
    For a circuit computing zero with top fan-in three, minimality is equivalent to the pairwise linear independence of $T_{1}, T_{2}, T_{3}$.
    In particular, none of $T_{1} + T_{2}, T_{2} + T_{3}, T_{1} + T_{3}$ are zero.
    This will be used implicitly throughout the rest of the argument.
    
    By construction, the sets $\cA, \cB, \cC$ each consist of non-associate forms.
    Suppose a form $A \in \cA$ and a form $B \in \cB$ are associate.
    Suppose $A, B$ are obtained by homogenising $A', B'$ respectively.
    Since $A, B$ are associate, the polynomials $A', B'$ are scalar multiples of each other.
    Since $A', B'$ are factors of $T_{1}, T_{2}$ respectively, and since $T = 0$, it must be that $A'$ is also a factor of $T_{3}$.
    This implies that $A' | \gcd\br{T_{1}, T_{2}, T_{3}}$, contradicting simplicity.
    Therefore $\cA$ and $\cB$ are disjoint, and the union $\cA \cup \cB$ consists of non-associate forms.
    A symmetric argument with the other pairs of sets shows that $\cA, \cB, \cC$ are pairwise disjoint, and that the union consists of non-associate forms.
    Since none of these sets contain $z$, the union $\bc{z}\cup \cA \cup \cB \cup \cC$ consists of pairwise non-associate forms.
    This shows the first condition.

    Now suppose $A \in \cA, B \in \cB$ are forms that are obtained by homogenising $A', B'$ respectively.
    We have $T_{3} \in \ideal{A', B'} \subset \radideal{A', B'}$.
    Let $\cC'$ be the set of irreducible factors of $\bc{P_{3j}}_{j}$, we have $\prod_{C' \in \cC'} C' \in \radideal{T_{3}} \subset \radideal{A', B'}$.
    Upon homogenisation, we obtain $\prod_{C \in \cC} C \in \radideal{zA, zB} \subset \radideal{A, B}$.
    Therefore $z \cdot \prod_{C \in \cC} C \in \radideal{A, B}$.
    This shows that the second condition for the configuration to be an EK configuration holds, and by symmetric arguments so do the third and fourth condition.
    
    We now show the final statement.
    The homogenisation of $P_{ij}$ is a polynomial of degree at most $d$ in a basis for the EK configuration.
    Therefore, these homogenisations span a vector space of dimension at most $\lambda^{d}$.
    Picking a basis and setting $z = 1$ gives a basis for the forms $P_{ij}$.
\end{proof}

Our results on rank bounds for $\tspsp$ circuits and deterministic PIT algorithm follows easily from our bounds on $EK$-configuration.
We restate the corollaries for convenience.

\rankbound*

\begin{proof}

    By \cref{def:circuit to config} we can define an EK configuration from the circuit $\Phi$.
    By \cref{theorem: EK main}, the rank of the configuration is bounded by $\lambda(d)$ for the function $\lambda$ defined in the theorem.
    By \cref{lem:tspsps is EK}, the rank of the circuit is bounded by $\rho(d) := \lambda(d)^{d}$.
\end{proof}

\mainPIT*

\begin{proof}
    By \cref{theorem: rankbound for identities}, simple minimal $\tspsp$ identities have rank bounded by $\rho(d)$.
    Therefore by \cite[Theorem~2]{beecken2013algebraic}, there is an algorithm for identity testing of such circuits that runs in time $\br{d s \rho(d) n}^{\bigO{d^{2} \rho(d)}}$, where $s$ is the size of the circuit.
\end{proof}

\section{Conclusion}\label{section: conclusion}

We prove rank bounds and therefore give the first polynomial time identity testing algorithm for the class of $\tspsp$ circuits.
We do so by showing that higher degree generalisations of Edelstein-Kelley configurations have bounded rank.
Our work builds upon the framework introduced in \cite{OS24} and refined in \cite{GOS24}.

The main open problem left by our work is to show rank bounds for $\spsp$ circuits, with $k > 3$.
This is the main conjecture in the work of \cite{gupta2014algebraic} and \cite{beecken2013algebraic} in the bounded bottom fan-in regime.
The main problem in this more general settings is that in such circuits, the same form can occur in more than one gate, even if the circuit is simple and minimal.
This increases the difficulty of the combinatorial steps in our proof.

\printbibliography

\end{document}